\documentclass[10pt,journal,cspaper,compsoc]{IEEEtran}
\usepackage{amsfonts}
\usepackage{algorithm}
\usepackage[footnotesize]{caption}
\usepackage[tight,footnotesize]{subfigure}

\ifCLASSOPTIONcompsoc
  \usepackage[nocompress]{cite}
\else
  \usepackage{cite}
\fi
\ifCLASSINFOpdf
  \usepackage[pdftex]{graphicx}
  \DeclareGraphicsExtensions{.pdf,.jpeg,.png}
\else
  \usepackage[dvips]{graphicx}
  \DeclareGraphicsExtensions{.eps}
\fi
\usepackage[cmex10]{amsmath}
\usepackage{algorithmic}
\usepackage{url}
\begin{document}
\title{Effective Clipart Image Vectorization
       Through Direct Optimization of Bezigons}

\author{Ming~Yang,
        Hongyang~Chao,
        Chi~Zhang,
        Jun~Guo,
        Lu~Yuan,
        and~Jian~Sun% <-this % stops a space
}


\markboth{IEEE Transactions on Visualization and Computer Graphics}%
{Yang \MakeLowercase{\textit{et al.}}: Effective Clipart Image Vectorization
Through Direct Optimization of Bezigons}

%\IEEEpubid{\makebox[\columnwidth]{
%\hfill 0000--0000/00/\$00.00~\copyright~2007 IEEE}%
%\hspace{\columnsep}\makebox[\columnwidth]{
%Published by the IEEE Computer Society\hfill}}

% PRIOR to the title within the \IEEEcompsoctitleabstractindextext IEEEtran
\IEEEcompsoctitleabstractindextext{%
\begin{abstract}
%\boldmath
Bezigons, i.e., closed paths composed of B\'ezier curves, have been widely
employed to describe shapes in image vectorization results. However, most
existing vectorization techniques infer the bezigons by simply approximating an
intermediate vector representation (such as polygons). Consequently, the
resultant bezigons are sometimes imperfect due to accumulated errors, fitting
ambiguities, and a lack of curve priors, especially for low-resolution images.
In this paper, we describe a novel method for vectorizing clipart images. In
contrast to previous methods, we directly optimize the bezigons rather than
using other intermediate representations; therefore, the resultant bezigons are
not only of higher fidelity compared with the original raster image but also
more reasonable because they were traced by a proficient expert. To enable such
optimization, we have overcome several challenges and have devised a
differentiable data energy as well as several curve-based prior terms. To
improve the efficiency of the optimization, we also take advantage of the local
control property of bezigons and adopt an overlapped piecewise optimization
strategy. The experimental results show that our method outperforms both the
current state-of-the-art method and commonly used commercial software in terms
of bezigon quality.
\end{abstract}
% not want either math or citations to appear in the abstract.

% Note that keywords are not normally used for peer review papers.
}

\maketitle

% is not selected <OR> if conference mode is selected - because compsoc
% conference papers position the abstract like regular (non-compsoc)
\IEEEdisplaynotcompsoctitleabstractindextext

% \fi
\IEEEpeerreviewmaketitle

\section{Introduction}
\label{sec:introduction}
% Computer Society journal papers do something a tad strange with the very
% first section heading (almost always called "Introduction"). They place it
% ABOVE the main text! IEEEtran.cls currently does not do this for you.
% However, You can achieve this effect by making LaTeX jump through some
% hoops via something like:
%
%\ifCLASSOPTIONcompsoc
%  \noindent\raisebox{2\baselineskip}[0pt][0pt]%
%  {\parbox{\columnwidth}{\section{Introduction}\label{sec:introduction}%
%  \global\everypar=\everypar}}%
%  \vspace{-1\baselineskip}\vspace{-\parskip}\par
%\else
%  \section{Introduction}\label{sec:introduction}\par
%\fi
%
% Admittedly, this is a hack and may well be fragile, but seems to do the
% trick for me. Note the need to keep any \label that may be used right
% after \section in the above as the hack puts \section within a raised box.

\IEEEPARstart{I}{mage} vectorization, also known as image tracing, is the
process of converting a bitmap image into a vector image. There are various
types of vectorization. In the present work, we focus on clipart image
vectorization. In such a case, the input raster is a clipart image, which is generally composed exclusively of digital illustrations like cartoons, logos, and symbols. Notably, this kind of images do not include photographs or scans of real hand-made drawings.

There is a huge demand for such a conversion technique. According to a survey
from \cite{diebel2008bayesian}, more than 7 million man hours are spent on
vectorizing images in the United States every year, and approximately 60\% of
the more than 10 million images to be vectorized are clipart images such as
logos and other rasterized vector art. As further evidence of the large demand
for clipart image vectorization, there is also a large market for online
services that specialize in tracing clipart images. The conversion can be
manually performed, but this may require a substantial amount of time and
effort, particularly for those users who are not proficient in tracing images.
This situation provides strong motivation for the development of an automated
algorithm for precise vectorization.

Notably, most modern methods that are appropriate for vectorizing clipart
images \cite{selinger2003potrace, weber2004autotrace, diebel2008bayesian,
zhang2009vectorizing} use bezigons to represent the resultant vector contours,
which has become the standard because of the compactness and editability of
bezigons.

However, almost no existing methods are specialized for directly obtaining
bezigons. Such methods typically direct most of their effort toward the
generation of intermediate polygons (Figure~\ref{fig:traditional-pipeline}b)
and consequently estimate bezigons (Figure~\ref{fig:traditional-pipeline}c)
that reproduce these polygons rather than the original image
\cite{selinger2003potrace,diebel2008bayesian,zhang2009vectorizing}. Among
these methods, \cite{diebel2008bayesian} (also known as Vector
Magic~\cite{vm2010}) generally produces the most accurate bezigon boundaries.
\footnote{In the case of vectorizing cartoon images with non-uniform color
regions or decorative lines,
  \cite{zhang2009vectorizing} may generate superior bezigons.} The key
to Vector Magic's success is that an effective and differentiable polygon-based
rasterization function was found, allowing polygon parameters to be precisely
optimized based on this function and polygon-specific priors. Nevertheless,
even this state-of-the-art method may still result in low-quality vectorized
bezigons (Figure~\ref{fig:traditional-pipeline}c), and other existing methods
are much more susceptible to such problems. There are three reasons for this
issue. First, errors introduced in the polygon estimation stage cannot  be
effectively corrected in the curve-fitting stage without observation of the
raster input. Second, even if the estimated polygons are perfect, ambiguities
still exist in the curve-fitting stage because of the nature of data
approximation. Third, bezigon-based priors have not yet been fully developed.
In short, generating bezigons in such an indirect manner may have a substantial
negative effect on the accuracy of the bezigon boundaries. This poses a serious
problem for clipart vectorization because even a slightly improper or
irrational boundary can be identified as a significant artifact in a clipart
image.

\begin{figure}[!t]
  \centering
  \subfigure[] {\includegraphics[height=1.1in]{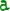}}
  \subfigure[] {\includegraphics[height=1.1in]{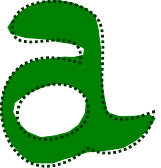}}
  \subfigure[] {\includegraphics[height=1.1in]{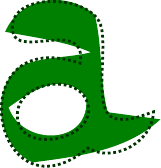}}
  \caption{Traditional pipeline of clipart vectorization. The dotted curves
  represent ground-truth outlines. (a) Raster input. (b) Intermediate
  representation (green polygons). (c) Final vector result (green bezigons).}
  \label{fig:traditional-pipeline}
\end{figure}

To solve the problems summarized above while retaining the advantages of the
state-of-the-art method~\cite{diebel2008bayesian}, an intuitive approach is to
devise an effective optimization mechanism that is specific to bezigons.

\begin{figure}[!t]
  \centering
  \subfigure[] {\includegraphics[width=1.6in]{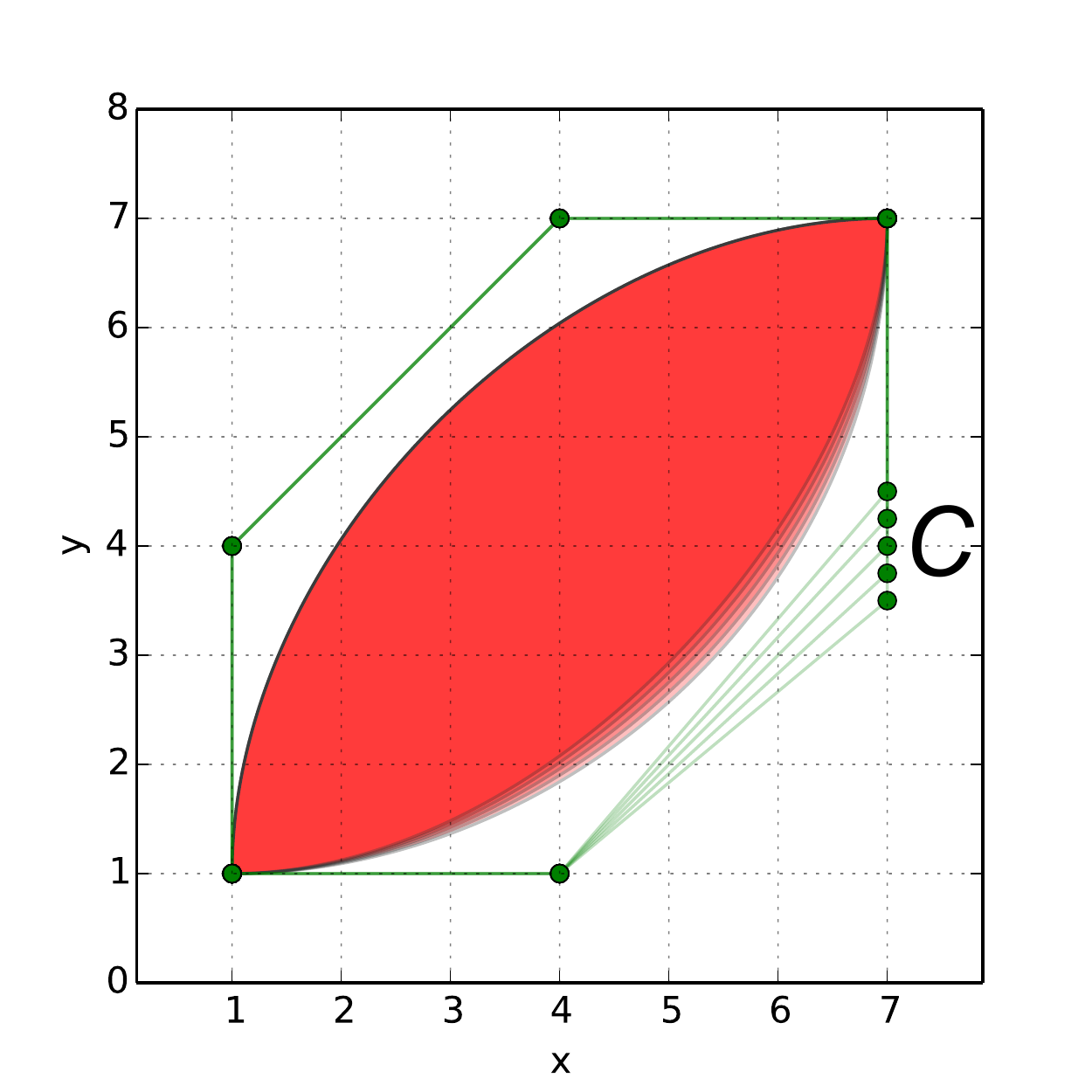}}
  \subfigure[] {\includegraphics[width=1.6in]{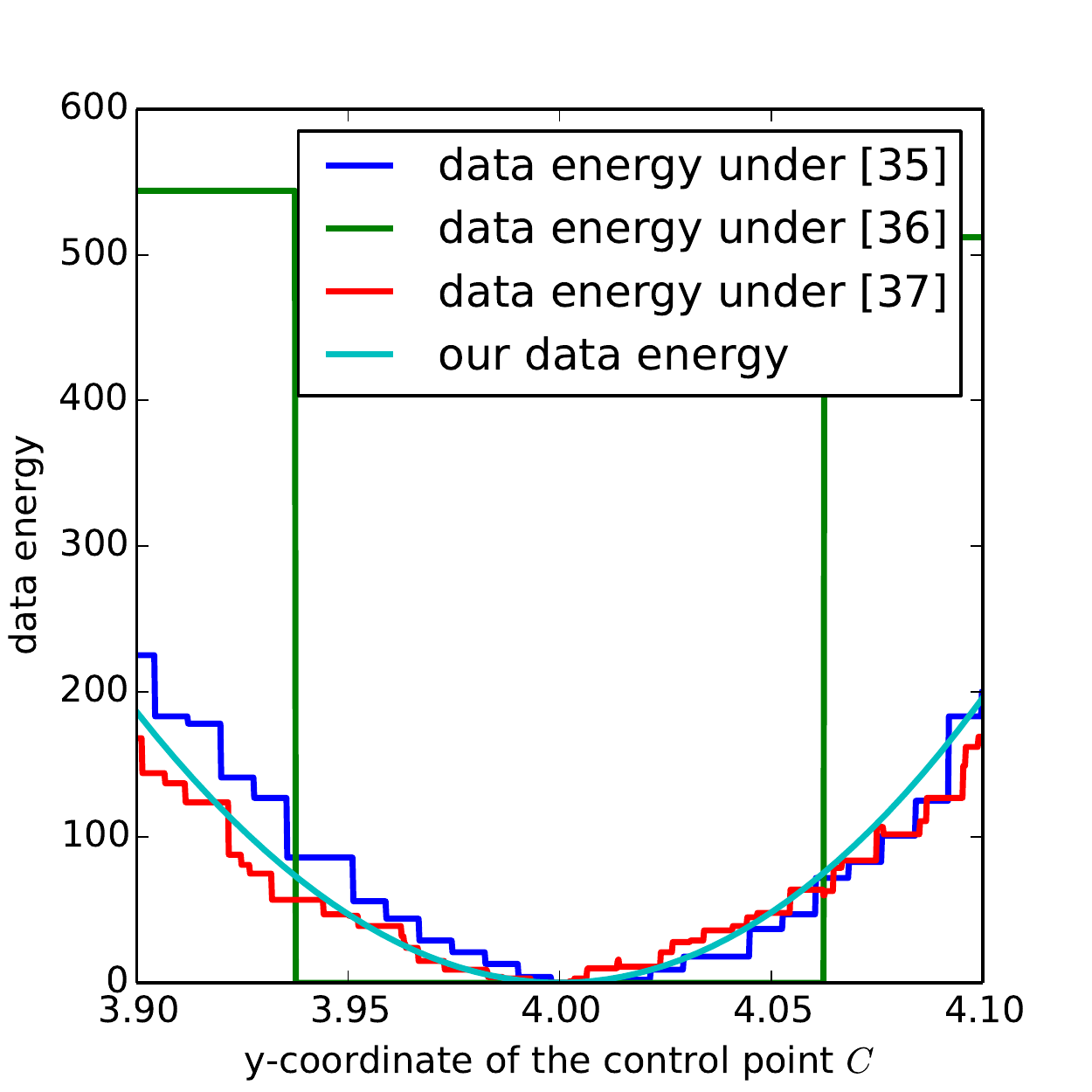}}
  \caption{Continuity of various candidate data energy functions. (a) Variation
  of a bezigon with the y coordinate of its control point $C$. (b) Variation of
  the data energy with the y coordinate of $C$ under various rasterization
  functions.}
  \label{fig:data-energy-comparison}
\end{figure}

However, establishing such a framework is non-trivial. In general, a direct
optimization of bezigons would necessitate an appropriate rasterization
function specialized for bezigons because such a function defines the bezigons'
fidelity to the raster image and serves as a fundamental basis for the entire
bezigon-specific optimization mechanism. However, most available rasterization
functions are not suited to this purpose because commonly used bezigon
rasterization methods are typically based on sampling sub-pixel locations of
the pixel grid~\footnote{In polygon-specific rasterizers, the bezigon is
approximated by a polygon before actual rasterization.}; the functions used in
these methods are non-differentiable, contain many discontinuities, and are
piecewise flat (have zero gradient with respect to the bezigon parameters)
almost everywhere (as shown in Figure~\ref{fig:data-energy-comparison}). These
properties impose a serious limitation on the effectiveness and efficiency of
the optimization procedure. Consequently, searching for a suitable bezigon-specific
rasterization function is the first challenge and the foremost problem that
must be overcome.

Even if this first challenge is overcome, the solution space might remain large
and contain many unreasonable bezigons that give rise to nearly the same raster
image (Figure~\ref{fig:toy-failcase} illustrates examples of such illegal
cases).  We observe that reasonable bezigons, when serving as vector
primitives, occupy only a small fraction of the parameter space of general
bezigons. There should be specific prior knowledge available regarding the
bezigons in typical vector images, and it is essential to incorporate such
prior knowledge to resolve the ambiguities and further constrain the solution
space. Unfortunately, little academic attention has been directly focused on
such prior knowledge; the available curve priors suggested in the literature
either cannot be directly applied for bezigons~\cite{diebel2008bayesian} or are
not specialized for vectorization~\cite{blake1998active}. Therefore, studying
the characteristics of both reasonable and unreasonable bezigons for image
vectorization, and incorporating closely related prior knowledge into our
bezigon optimization, is another challenge to be addressed.

\begin{figure}[!t]
\centering
  \subfigure[] {\includegraphics[width=0.7in]{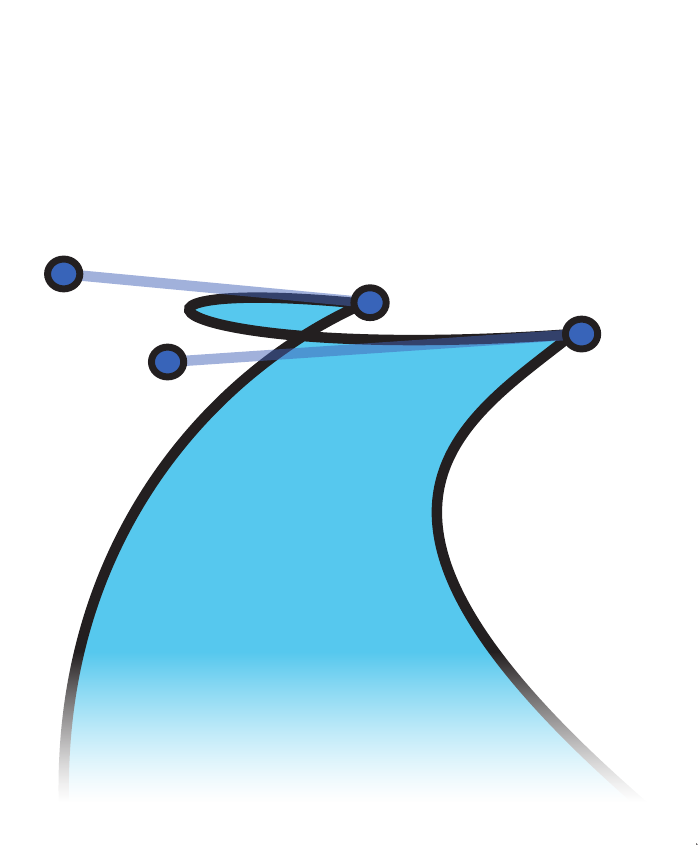}}
  \subfigure[] {\includegraphics[width=0.7in]{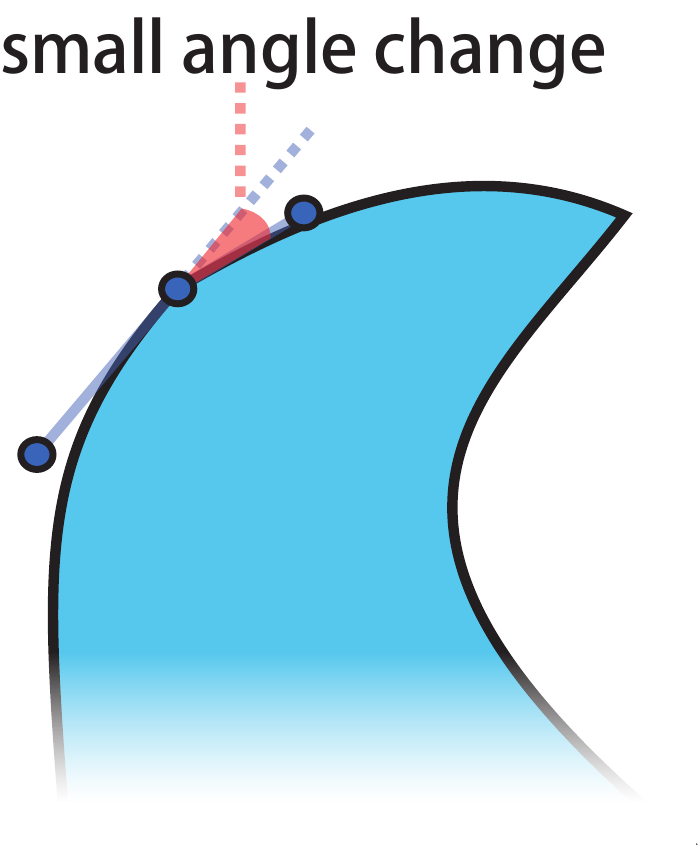}}
  \subfigure[] {\includegraphics[width=0.7in]{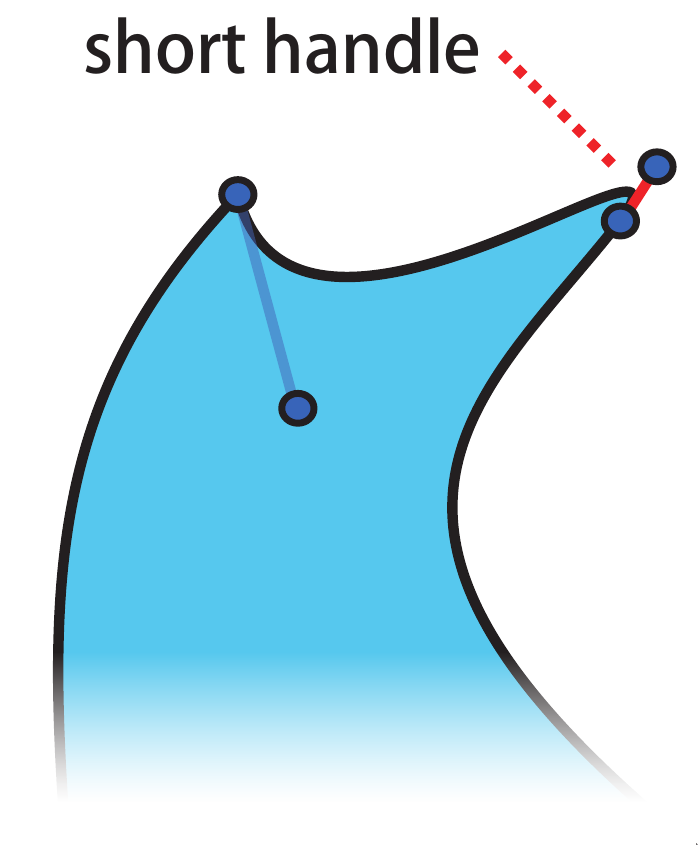}}
  \subfigure[] {\includegraphics[width=0.7in]{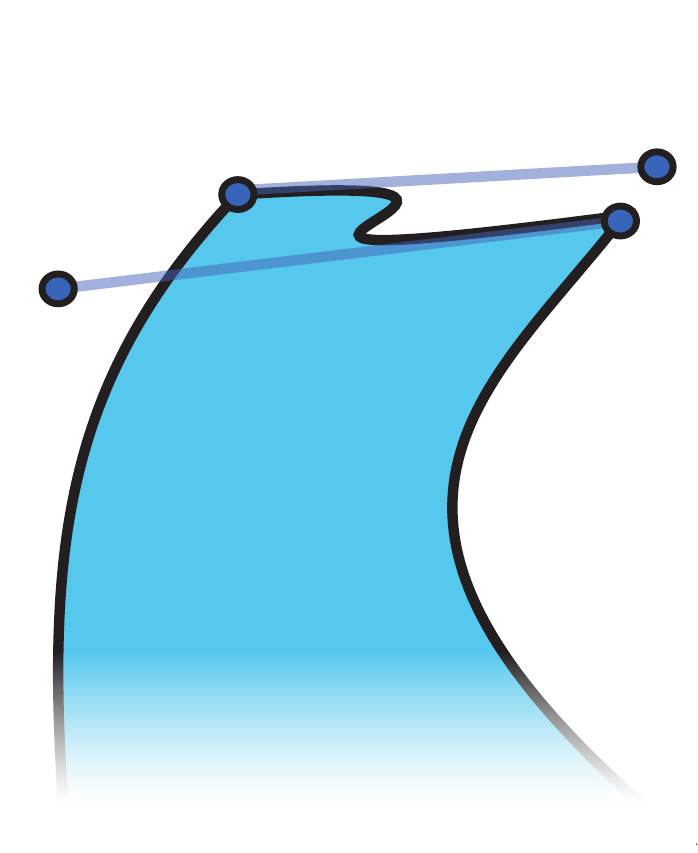}}
  \caption{Four types of failure cases that occur when only data energy is
  considered. (a) Self-intersection. (b) False corners with small angle
  variations. (c) Short handle. (d) Twisted section.}
  \label{fig:toy-failcase}
\end{figure}

In this paper, we present solutions to the above challenges and propose a
bezigon-specific optimization framework for more precise clipart vectorization.
Our main contributions are as follows:
\begin{itemize}
\item By analyzing several rasterization approaches, we identify an
  appropriate rasterization function, theoretically prove certain analytic
  properties thereof that facilitate effective optimization for our purposes
  (using the theory of \emph{generalized functions}
  \cite{gel1968generalized}), and experimentally validate its
  compatibility and robustness for vectorizing various types of clipart images.
  Thus, we establish a framework for clipart vectorization via the direct
  optimization of bezigons. Meanwhile, we provide some approximate criteria for
  determining whether a rasterization function is suitable for optimization-
  based image vectorization.
\item Based on an intensive study of reasonable bezigons in typical
  vector images as well as unreasonable bezigons arising from experiments, we
  classify the common illegal cases of bezigon primitives into four categories:
  self-intersections, false corners with small angle variations, short handles,
  and twisted sections (Figure~\ref{fig:toy-failcase}). To address these
  illegal cases, we suggest a self-intersection prior term, an angle-variation
  prior term, a B\'{e}zier-handle prior term, and a curve-length prior term.
  All these terms are incorporated into our framework to further constrain the
  solution space and to provide broadly reasonable guidance for bezigon
  optimization. Moreover, errors in the curve boundaries, if any remain, become
  visually insignificant because the resultant bezigons are more reasonable and
  aesthetically pleasing in general.
\item By taking full advantage of the local control property of
  B\'ezier curves, we propose a piecewise optimization strategy to effectively
  solve the problem of bezigon optimization. This strategy considerably reduces
  the computational cost and makes our vectorization method more practical.
\item Based on the above techniques, we suggest a new bezigon
  optimization framework. In such a framework, we can effectively vectorize a
  clipart image or refine vector results obtained using other approaches.
  Notably, such a framework is generally capable of incorporating any bezigon
  rasterization model and additional prior knowledge for the purpose of image
  vectorization or other applications, such as curve stylization.
\end{itemize}

The experimental results show that our method outperforms both the current
state-of-the-art method and commonly used commercial software in terms of
bezigon quality, especially in tough vectorization cases such as smooth
boundaries with high curvatures, obtuse corners, and slightly bent edges.

The remainder of this paper is organized as follows:
Section~\ref{sec:related-work} briefly reviews existing clipart vectorization
approaches. Section~\ref{sec:overview} formulates the problem of clipart
vectorization in terms of bezigon optimization. An overview of the proposed
vectorization framework and our points of focus is also provided in this
section. Section~\ref{sec:directly-optimizing-bezigon} fully explains our
approach to the direct optimization of bezigons for image vectorization. The
experimental results and comparisons are presented in
Section~\ref{sec:experiments}, and the paper concludes with a discussion of
further perspectives on this work in Section~\ref{sec:conclusion}.

\section{Related Work}
\label{sec:related-work}

Various other types of image vectorization methods exist that are specific to
line drawings
\cite{fahn1988topology, nagasamy1990engineering, hori1993raster, chiang1995new,
dori1999sparse, hilaire2006robust, noris2013topology}, natural images
\cite{swaminarayan2006rapid, lecot2006ardeco, xia2009patch, price2006object,
sun2007image, lai2009automatic, orzan2013diffusion, jeschke2009rendering,
olsen2011image}, and pixel art \cite{kopf2011depixelizing}. However, these
methods merely capture the intrinsic nature of clipart images and are likely to
fail in generating precise curve boundaries; thus, they are not well suited for
the task considered here.

In the last decade, several methods \cite{sykora2005sketching,sykora2005video,
diebel2008bayesian, zhang2009vectorizing, selinger2003potrace} have been
proposed for clipart image vectorization. These methods typically involve
segmenting the input image into a set of regions and inferring the color and
the boundary location for each region.

To overcome the poor quality of the segmentation that results from general
image vectorization, \cite{sykora2005sketching} exploited a visual feature
of certain types of cartoons, i.e., shapes that are typically bounded by bold
dark contours, and succeeded in producing a more precise segmentation technique
for clipart images. However, this approach could only address regions enclosed
by such strokes, which is not always the case in modern clipart images.

To further improve the segmentation and more semantically infer the shape
color, \cite{zhang2009vectorizing} proposed a novel trapped-ball segmentation
method that can segment a clipart image more semantically even when some
regions are non-uniformly colored. Moreover, this approach considers temporal
coherence and is capable of vectorizing cartoon animations. Such progress is
impressive, but segmentation, color estimation and vectorizing animations are
not our topics of focus.

Perhaps the most difficult aspect of image vectorization still lies in the
inference of boundary locations. As previously mentioned,
\cite{diebel2008bayesian} is the state-of-the-art vectorization
algorithm with respect to its precision of boundary location, especially for
the vectorization of uniformly colored shapes. However, the contour
optimization of this method, which plays the most important role in the
algorithm, is specialized for polygons rather than bezigons and hence
occasionally results in inaccurate bezigons. It seems that extending this
method's approach to curve fitting by somehow managing to fully use the
information provided by the raster input might solve the problem. However, this
is a non-trivial task for the reasons mentioned in
Section~\ref{sec:introduction}. Moreover, this process would result in a
bezigon optimization problem similar to ours.

In addition to the academic literature, there are also a number of related
commercial tools, such as Adobe Illustrator~\cite{adobe2013}, Corel
CorelDRAW~\cite{corel2013} and Vector Magic~\cite{vm2010} (a product based on
the technology of \cite{diebel2008bayesian}), as well as open-source projects
such as Potrace \cite{selinger2003potrace} and AutoTrace
\cite{weber2004autotrace}. Of these tools, Adobe Illustrator is the most
representative, and Vector Magic achieves the best results in terms of bezigon
boundary precision. In this paper, we compare our algorithm with these two
software packages. Although the technical details of most commercial tools are
unavailable, the experimental results indicate that these tools exhibit a
problem similar to (or even worse than) that of
\cite{diebel2008bayesian, vm2010}.

In summary, insufficient precision in identifying bezigon boundaries is the
most common shortcoming of existing vectorization methods. Therefore, improving
the precision of bezigon boundaries, which is important for vectorizing clipart
images, is the primary goal of this paper.

\section{Problem Formulation and Overview of Our Framework}
\label{sec:overview}

To facilitate a better understanding of this paper, in this section, we
formulate the related problem along with the relevant notation and then provide
an overview of the proposed vectorization framework and our topics of interest.

For the sake of simplicity, we consider only a single bezigon. Our work can
easily be extended to situations that involve two or more bezigons because each
bezigon can be independently vectorized.

\subsection{Problem Formulation}

Given a raster image, the primary task of clipart image vectorization is to
infer a bezigon from the raster input. In a typical vector image, a bezigon can
be completely determined by its \emph{geometric parameters} and its \emph{color
parameters}.

\textbf{Geometric parameters.} As previously mentioned, a bezigon $S(t)$ is
simply a series of B\'ezier curves joined end to end, i.e.,
\begin{equation}
S(t)=
\begin{cases}
B_1(t), & t \in [0,1], \\
B_2(t-1), & t \in [1,2], \\
\qquad \qquad \qquad \vdots \\
B_N(t-N+1), & t \in [N-1,N].
\end{cases}
\end{equation}

Here, $N$ denotes the number of curves in the bezigon, and $B_j(t)$ represents
the $j$-th curve, which is assumed without loss of generality to be a 2D cubic
B\'ezier curve with the following parametric form:
\begin{equation}
\label{eq:X_j-Y_j}
\begin{array}{rll}
X_j(B_x;t) &=& \sum_{i=0}^3 {\binom{3}{i}}(1 - t)^{3 - i}t^i x_{j,i+1}, \\
Y_j(B_y;t) &=& \sum_{i=0}^3 {\binom{3}{i}}(1 - t)^{3 - i}t^i y_{j,i+1}, \\
j &=& 1,2,\hdots,N, \\
t &\in& [0,1],
\end{array}
\end{equation}
where the $(x_{j,i},y_{j,i}) \in \mathbb{R}^2 (i=1,2,3,4; j=1,2,\hdots,N)$
constitute the four control points of the $j$-th B\'ezier curve. The last
control point of one curve coincides with the starting point of the next curve,
i.e., $x_{j,4}=x_{j+1,1},y_{j,4}=y_{j+1,1} (j=1,2,\hdots,N)$. Thus, all
geometric parameters $B$ of a bezigon can be represented as
\begin{equation}
\label{eq:B_x-B_y}
\begin{array}{c}
B=(B_x,B_y) \\
B_x \!\!=\!\! \begin{pmatrix}
x_{1,1} \!\! & \!\! x_{1,1} \!\! & \!\! x_{1,1} \\
x_{2,1} \!\! & \!\! x_{2,1} \!\! & \!\! x_{2,1} \\
\vdots \!\! & \!\! \vdots \!\! & \!\! \vdots \\
x_{N,1} \!\! & \!\! x_{N,1} \!\! & \!\! x_{N,1}
\end{pmatrix},
B_y \!\!=\!\! \begin{pmatrix}
y_{1,1} \!\! & \!\! y_{1,1} \!\! & \!\! y_{1,1} \\
y_{2,1} \!\! & \!\! y_{2,1} \!\! & \!\! y_{2,1} \\
\vdots \!\! & \!\! \vdots \!\! & \!\! \vdots \\
y_{N,1} \!\! & \!\! y_{N,1} \!\! & \!\! y_{N,1}
\end{pmatrix}.
\end{array}
\end{equation}

\textbf{Color parameters.} Without loss of generality, we consider that the
color of the bezigon at pixel $(x,y)$ is represented by the function
$c(C;x,y)$. If the region color is assumed to be uniform, then $c(C;x,y)=C_0$,
and the color parameter is $C_0=(r_0,g_0,b_0) \in \mathbb{R}^3$. If a quadratic
color model is assumed, then $c(C;x,y)=C_0+C_1 x+C_2 y+C_3 x^2+C_4 xy+C_5 y^2$;
thus, the color parameters are $C=(C_0,C_1,C_2,C_3,C_4,C_5) \in
\mathbb{R}^{18}$.

Now, for a given raster input image $I$, our objective can be considered to be
the inference of the parameters
\begin{equation}
W=(B,C)
\end{equation}
from $I$ such that the bezigon that is defined by $W$ can explain the input
image $I$. In other words, the raster image of the bezigon should be similar to
the input image. The problem is obviously a typical non-linear and ill-posed
problem because there may be many possible solutions because of uncertainties
in the imaging process and ambiguities of visual interpretation. To resolve the
intrinsic ill-posedness of the problem, we must further constrain the solution
space by introducing additional prior knowledge regarding bezigons in vector
images.

Based on the above discussion, we will adopt an energy minimization approach
that is widely used in many computer vision algorithms
\cite{treiber2013optimization}.

We first define our energy function as
\begin{equation}
\label{eq:E}
E(W;I)=E_{data}(W;I)+E_{prior}(W),
\end{equation}
where $E_{data}(W;I)$ is the so-called \emph{data energy}, which measures the
fidelity of a vector solution to the observed raster image, and $E_{prior}(W)$
is the so-called \emph{prior energy}, which is the formulation of our
constraints or prior knowledge regarding reasonable bezigons for the above-
mentioned vector images.

Consequently, the problem of this paper will be formulated in terms of
identifying the optimal bezigon $W^*$ such that
\begin{equation}
\label{eq:w*}
W^*=\arg \min_{W \in \Omega}⁡(E(W;I)).
\end{equation}

\subsection{Overview of Our Framework for Optimization}

Once our energy function is fully specified, the entire energy minimization
framework can be divided into two phases: a bezigon initialization phase and a
bezigon-specific optimization phase (Figure~\ref{fig:overview}).

\begin{figure}[!t]
\centering
\includegraphics[width=3.3in]{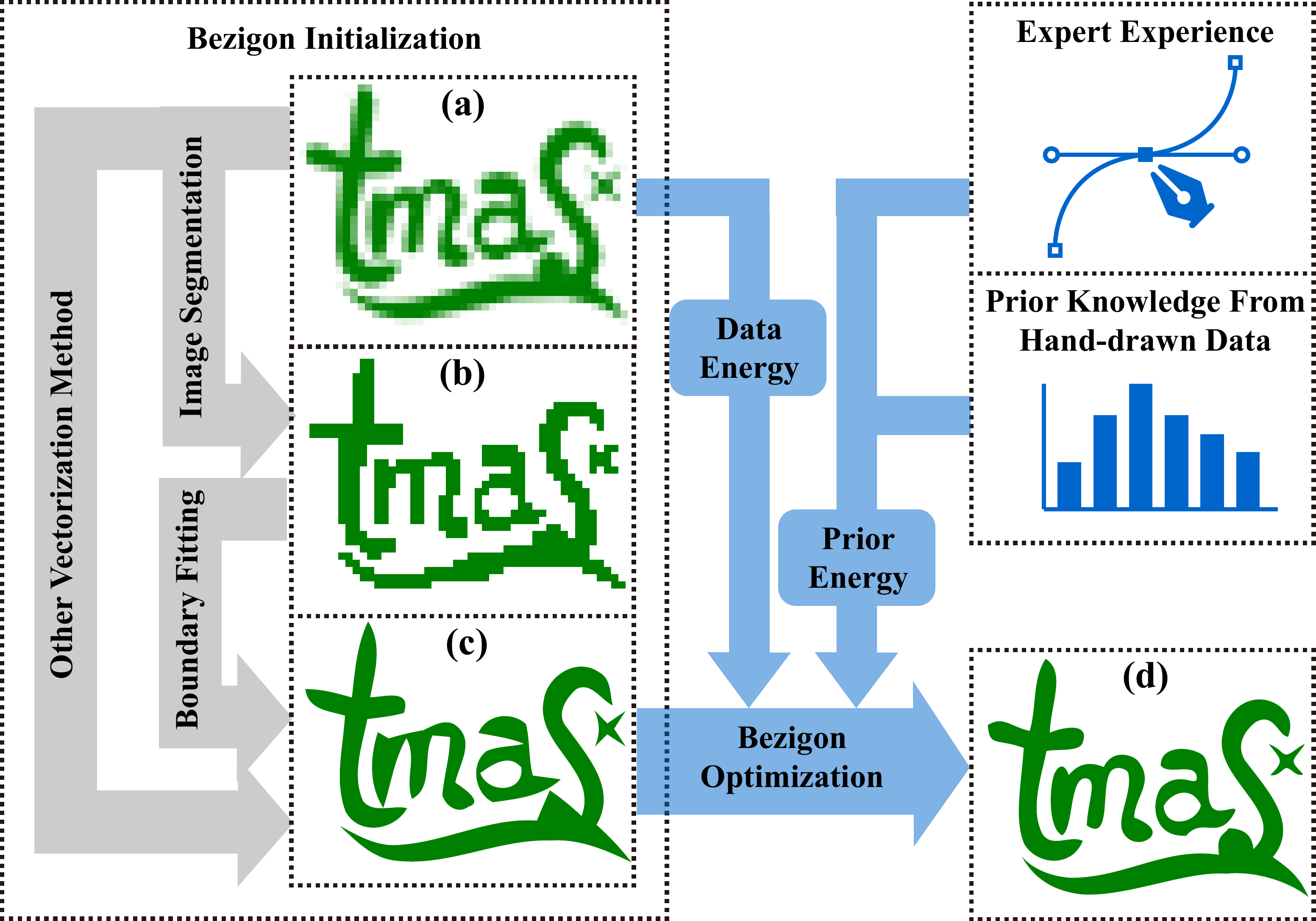}
\caption{Overview of our framework. (a) Raster input. (b) Segmentation result.
(c) Initial bezigons. (d) Optimized bezigons.}
\label{fig:overview}
\end{figure}

\textbf{Bezigon initialization phase.} The initialization phase takes a raster
image (Figure~\ref{fig:overview}a) as input and outputs a set of initial
bezigons (Figure~\ref{fig:overview}c). These bezigons can be either obtained
using other existing vectorization methods or extracted from the input image. A
simple, fully automated method of accomplishing this extraction consists of two
steps: a segmentation step that is used to segment the input image into a set
of regions
\cite{felzenszwalb2004efficient} (Figure~\ref{fig:overview}b) and a
boundary-fitting step to fit a piecewise cubic B\'ezier curve to the boundary
of each region \cite{schneider1990algorithm} (Figure~\ref{fig:overview}c). As
another option, the initial bezigons can also be manually drawn or
interactively refined by the user. Regardless of which method is used, the
obtained bezigons serve as initial parameters in the next phase; hence, they
are not necessary highly accurate. The technical details of this phase are
outside of the scope of this paper.

\textbf{Bezigon optimization phase.} The optimization phase is the primary
concern of this paper. This phase includes direct bezigon optimization, which
is the task that we are emphasizing. This process takes the initial bezigons
from the first phase as input and outputs the optimal bezigons as the final
vector result. In contrast to other existing vectorization approaches, this
phase of our framework consists of neither simply applying a curve-fitting
algorithm (e.g., \cite{schneider1990algorithm}) nor indirectly optimizing
bezigons according to an intermediate representation (e.g., polygons in
\cite{diebel2008bayesian}). Instead, we optimize the bezigon parameters by
directly observing the raster input and incorporating both the image-tracing
experience of experts and prior knowledge from existing hand-drawn vector
images. The bezigon optimization and these sources of information are
simultaneously bridged by our data energy and prior energy. In this way,
unnecessary accumulated errors introduced by the intermediate process can be
avoided, and hence, the quality of the resultant bezigon can be improved.
However, as stated previously, there are several as- yet-unresolved challenges
arising from such an optimization approach. Therefore, our paper will emphasize
these issues. Section~\ref{sec:directly-optimizing-bezigon} presents a
discussion of our solution method and explains our main contributions.

There are three major advantages to our framework. First, the error arising
from the entire vectorization pipeline can be minimized. Second, any bezigon-
based priors can be conveniently incorporated to generate even more reasonable
results, once we have a better understanding of bezigons in typical vector
images, or to cause the resultant bezigons to satisfy certain constraints of
other specific applications. Third, our vectorization approach behaves
similarly to bezigon evolution, which is particularly well suited to the
vectorization of clipart animations, and facilitates the further refinement of
inaccurate bezigons resulting from other vectorization approaches.

\section{Approach for Directly Optimizing Bezigons}
\label{sec:directly-optimizing-bezigon}

In this section, we solve some key issues related to bezigon
optimization. The optimization involves specifying the data energy with the
proper rasterization model (Section~\ref{subsec:data-energy}) and
several bezigon-specific prior terms
(Section~\ref{subsec:prior-energy}). To more efficiently solve
Equation~\ref{eq:w*}, we also explore the nature of bezigon parameters
and propose a piecewise optimization strategy
(Section~\ref{subsec:piecewise}). In the following, we will use the
same notations as are used in Section~\ref{sec:overview}.

\subsection{Data Energy}
\label{subsec:data-energy}

To fully utilize the information provided by the input image, we define the
data energy as the distance between the input image $I$ and the image generated
by rasterizing a vector solution $W$:
\begin{eqnarray}
E_{data}(W;I) \!\!\! &=& \!\!\! \frac{1}{l_0} \|R(W)-I\|^2 \\
              \!\!\! &=& \!\!\! \frac{1}{l_0} \sum_{(x,y) \in \Lambda}
              \|R(W;x,y)-I(x,y)\|^2.
\label{eq:data}
\end{eqnarray}
Here, the function $R(W)$ models a specific bezigon rasterization process. The
function takes the parameters W of a bezigon as input and produces a raster
image of the same size as the input image $I$. $R(W;x,y)$ and $I(x,y)$ denote
the values at pixel $(x,y)$ in the rasterized image given by $R(W)$ and in the
input image, respectively. $\Lambda$ is the lattice of the input image $I$. The
denominator $l_0$ represents the arc length of the initial bezigon. This
denominator is fixed during the optimization and can be easily estimated from
the geometric parameters $B^0=(B_x^0,B_y^0)$ of the initial bezigon, i.e.,
$$l_0 = \sum_{j=1}^N \int_0^1 \sqrt{\left(\frac{dX_j(B_x^0,t)}{dt}\right)^2
      + \left(\frac{dY_j(B_y^0,t)}{dt}\right)^2} dt.$$

Two issues now arise for consideration. First, a bezigon rasterization function
for $R(W;x,y)$ should be specified because such a function is essential to make
Equation~\ref{eq:data} suitable for optimization. It is also one of the most
challenging aspects of direct bezigon optimization. As previously mentioned,
the most important contribution of the current state-of-the-art approach
\cite{diebel2008bayesian} also lies in finding an appropriate rasterization
function, but one that is specific to polygon optimization. For bezigon
optimization, however, research concerning suitable rasterization functions is
still lacking in the existing literature. Second, we must address the case in
which the input image is not generated by the specified rasterization function
used for the data energy because the rasterization method that generates the
given input image is generally unknown and most likely not the same as our
function.

Regarding the first issue,  several methods exist for directly or indirectly
rasterizing bezigons~\cite{catmull1978hidden, duff1989polygon,
doan2004antialiased, manson2011wavelet, manson2013analytic, cairo2013,
batik2002, agg2006}, each of which corresponds to a candidate rasterization
function $R(W;x,y)$. However, we find that nearly all such functions yield poor
results when a typical solver for nonlinear optimization (such as conjugate
gradient, l-BFGS, or NEWUOA) is applied. This is because most available
rasterization functions are either piecewise flat, or discontinuous, almost
everywhere (as shown in Figure~\ref{fig:data-energy-comparison}b). Although
such discontinuities pose no problems for common rasterization tasks, they can
strongly degrade the effectiveness or efficiency of optimization. Various
specific optimization approaches (such as \cite{ermoliev1995minimization}) can
be applied in the case of discontinuous functions. However, our experimental
results indicate that such approaches often fail to produce satisfactory
bezigons. Moreover, these solvers are relatively slow, which limits their use
in image vectorization. Based on the above experiments and analysis, we
recognize that an appropriate rasterization function should exhibit certain
properties, such as continuity with respect to the bezigon parameters.
Moreover, if the rasterization function is also differentiable with respect to
those parameters, more efficient and effective solvers can be adopted to
optimize our energy function to obtain better results.

In the search for proper rasterization approaches, the approach presented in
\cite{manson2011wavelet} came to our attention. This approach uses a
hierarchical Haar wavelet representation to analytically calculate an
anti-aliased raster image of bezigons. According to \cite{manson2011wavelet},
for a bezigon $W$, the pixel color value at $(x,y)$ in the resultant raster
image can be calculated as follows:
\begin{equation}
\label{eq:R_MS}
\begin{split}
&R_{MS}(W;x,y) = \\
&c(C;x,y) \sum_{j=1}^N
        \left\{\begin{split}
        \sum_{k \in K} c_{0,k}^{(0,0)} (B;j) \psi_{0,k}^{(0,0)} (x,y) \\
        + \sum_{s=0}^d \sum_{k \in K}
                \left[\begin{split}
                c_{s,k}^{(0,1)} (B;j) \psi_{s,k}^{(0,1)} (x,y) \\
                +c_{s,k}^{(1,0)} (B;j) \psi_{s,k}^{(1,0)} (x,y) \\
                +c_{s,k}^{(1,1)} (B;j) \psi_{s,k}^{(1,1)} (x,y)
                \end{split}\right]
        \end{split}\right\}
\end{split}
\end{equation}

Here, $s$ represents a specific scaling from the original resolution to the
pixel solution $d$, and $k=(k_x,k_y)$, $k_x \in K_x \subset \mathbb{Z}$, $k_y
\in K_y \subset \mathbb{Z}$ represents a specific translation in the finite set
$K=K_x \times K_y \subset \mathbb{Z}^2$ corresponding to all possible
translations in the current scaling. $\psi_{s,k}^{(\cdot)} (x,y)$ and
$c_{s,k}^{(\cdot)} (B;j)$ are a two- dimensional Haar wavelet basis function
and its coefficient, respectively. The definitions of these two functions can be
found in Appendix~\ref{subsec:appendix-definitions}.

Although \cite{manson2011wavelet} provides a closed-form solution for
rasterizing bezigons, the continuity and differentiability of $R_{MS} (W;x,y)$
are not obvious because of the discontinuity of the Haar wavelet basis
functions. One of the most important tasks of this section is to present the
proofs of the continuity and differentiability of this rasterization function.
The latter is not straightforward. To obtain the proof, we must rely on several
properties and operations from the theory of \emph{generalized functions}
\cite{gel1968generalized}.

Note that for any given coordinate $(x,y)$, $R_{MS} (W;x,y)$ is a function of
the bezigon parameters $W$. To establish the function's continuity and
differentiability, we present the following theorems.

\newtheorem{theorem}{Theorem}
\begin{theorem}[continuity]
$R_{MS} (W;x,y)$ is a continuous function with respect to all bezigon
parameters $W$.
\end{theorem}

\begin{proof} As previously stated, the bezigon parameters $W$ consist of the
color parameters $C$ and the geometric parameters $B$. According to
Equation~\ref{eq:R_MS}, $R_{MS} (W;x,y)$ is continuous as long as the assumed
color model $c(C;x,y)$ is continuous with respect to the color parameters $C$,
which is often the case. With respect to the geometric parameters $B$, $R_{MS}
(W;x,y)$ is also continuous. A detailed analysis can be found in
Appendix~\ref{subsec:appendix-continuity}.
\end{proof}

Obviously, if $R_{MS} (W;x,y)$ serves as our rasterization function $R(W)$,
then the resultant data energy is also a continuous function. The smooth curve
that corresponds to our data energy in Figure~\ref{fig:data-energy-comparison}b
reflects such a property as well. The continuity of the data energy not only
enables us to apply a common solver for the nonlinear optimization but also
facilitates the resolution of any ambiguity that arises from the observation of
the input data.

\begin{theorem}[differentiability]
\label{theorem:differentiability}
$R_{MS} (W;x,y)$ is differentiable with respect to the bezigon parameters $W$.
\end{theorem}

\begin{proof} Most color models $c(C;x,y)$ are differentiable with respect to
the color parameters $C$. In such cases, $R_{MS} (W;x,y)$ is obviously
differentiable with respect to the color parameters, according to
Equation~\ref{eq:R_MS}. However, the differentiability of $R_{MS} (W;x,y)$ with
respect to the geometric parameters $B$ is not obvious. We use the theory of
generalized functions to analyze this matter. Because of space limitations and
the complexity of the discussion, the proof and the derivatives are presented
in Appendix~\ref{subsec:appendix-derivatives}.
\end{proof}

Based on the above analysis and theorems, we can conclude that $R_{MS} (W;x,y)$
may be a suitable candidate for the rasterization function $(W;x,y)$ in
Equation~\ref{eq:data}. Therefore, this rasterization function may be adopted
in the proposed framework. Then, our final data energy can be rewritten as
\begin{equation}
E_{data}(W;I) = \sum_{(x,y) \in \Lambda} \|R_{MS}(W;x,y)-I(x,y)\|^2.
\end{equation}

Now, we consider the second issue. Because there are many commonly used
rasterization methods, it is often the case that the input raster image is not
generated by the rasterization method used in our data energy term. This could
be an issue if there are significant differences in the rasterization results
between our chosen method and the method used to generate the input image.
Therefore, to ensure the practical utility of the proposed vectorization
method, we must investigate whether the selected rasterization function can
closely approximate the rendering results of other commonly used rasterization
approaches.

Fortunately, our selected rasterizer $R_{MS} (W;x,y)$ is still a suitable
choice in this context. To prove this claim, we perform the following
experiment: We collect a set of real-world vector images. All these vector
images are rasterized by each of the commonly used anti-aliased rasterizers,
using the recently proposed methods, and by $R_{MS} (W;x,y)$. Note that the
only possible differences in images produced by different rasterizers lie in
pixels that intersect the bezigon boundary. To further clarify the comparison,
we consider only the differences among such pixels in the resultant images.
Histograms of these differences are presented in
Figure~\ref{fig:data-assumption}. It is readily apparent that a large
proportion of the ``boundary'' pixels that are rendered by any other rasterizer
remain identical those produced by $R_{MS} (W;x,y)$. Moreover, all
distributions have means of zero and small variances. Therefore, the pixel
values generated by our rasterization function can be safely assumed to be a
good approximation to those generated by other commonly used rasterization
methods, and hence, our rasterization function can still accurately model the
original rasterization process of most clipart images.

\begin{figure}[!t]
\centering
\includegraphics[width=2.5in]{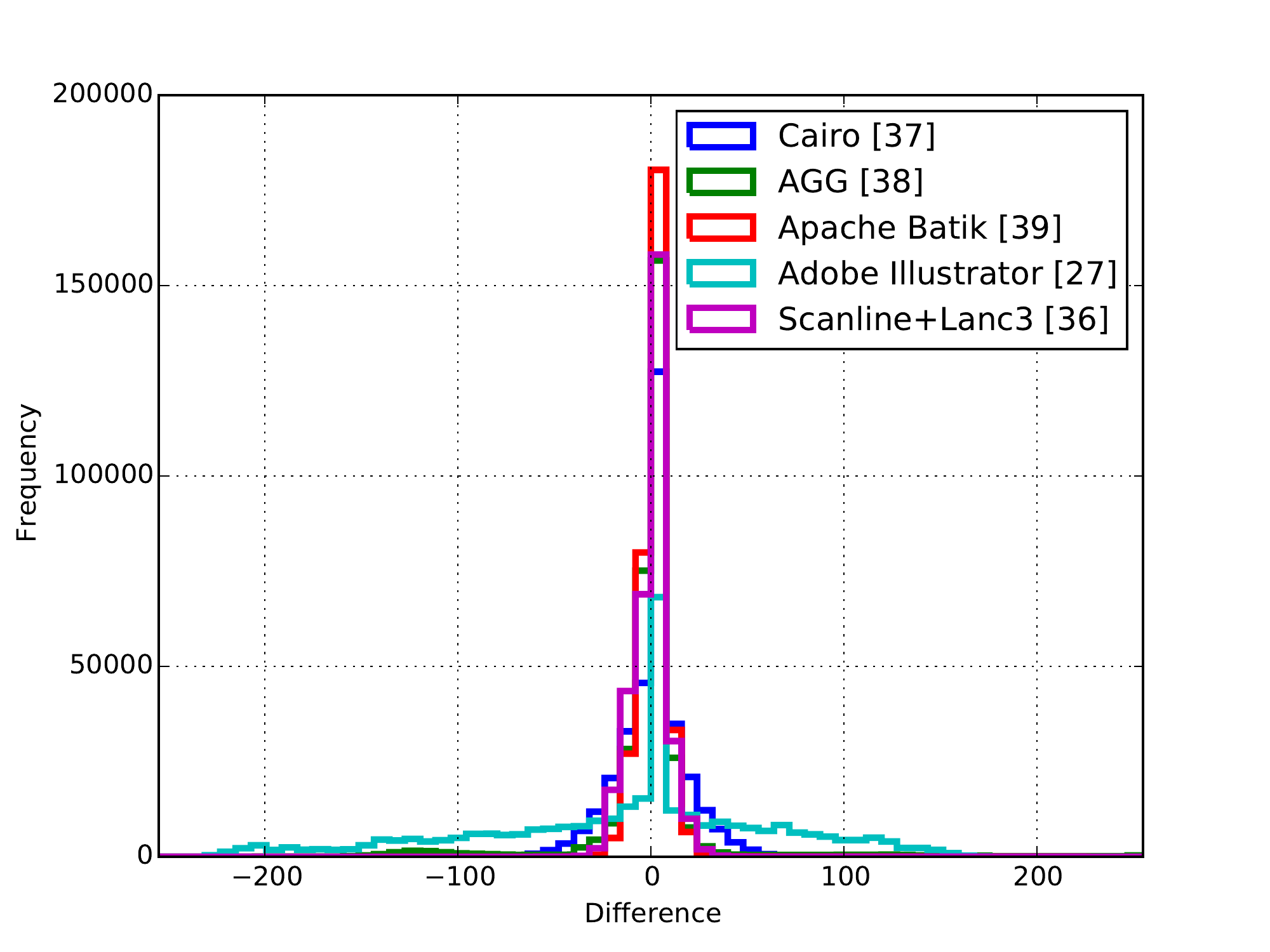}
\caption{Histograms of differences between pixel values produced by $R_{MS}
(W;x,y)$ and those produced by the other rasterizers listed in the figure.}
\label{fig:data-assumption}
\end{figure}

In summary, we have proven the suitability of our bezigon rasterization
function for optimization as well as its compatibility with various clipart
raster input, and we have presented the definition of our data energy. Notably,
for any other rasterization function that is a candidate for application to
vectorize a certain type of image, a similar procedure should be followed to
evaluate the suitability and compatibility of that function.

\subsection{Prior Energy}
\label{subsec:prior-energy}

After the data energy has been carefully selected, various simple cases (e.g.,
the vectorization of a simple bezigon in a high-resolution raster image) can
already be effectively addressed when there is adequate information implicit in
the observed raster data. However, it is more often the case that the bezigons
are relatively complex and that the information available in the raster input
is inadequate. In such a case, profound uncertainty regarding the correct
solution may remain if the data energy alone is considered. Therefore, the
optimization may result in unreasonable bezigons that can be easily identified
by the human eye.

Indeed, our intensive experiments provide evidence of such issues. More
specifically, the failure cases of direct bezigon optimization using only data
energy generally fall into four categories: (a) self-intersections, (b) false
corners with small angle variations, (c) short handles, and (d) twisted
sections (Figure~\ref{fig:toy-failcase}).

All these bezigons are considered to be unreasonable because they are
aesthetically unpleasing and, according to expert opinion, are unlikely to be
drawn or traced by a professional illustrator. These types of bezigons are also
rare in typical vector images. (Taking self-intersection as an example, we find
that very few bezigons in vector images from the Open Clipart
library~\cite{phillips2005introduction} intersect with themselves. Most
bezigons that exhibit self-intersection are believed to have been created by an
amateur or automatically traced from a raster image.) The reason for the
occurrence of such illegal bezigons is that their corresponding raster images
are quite similar to the input images (compare
Figures~\ref{fig:failcase-intersect}f, \ref{fig:failcase-anglechange}f,
\ref{fig:failcase-handle}f and \ref{fig:failcase-short}f with
\ref{fig:failcase-intersect}h, \ref{fig:failcase-anglechange}h,
\ref{fig:failcase-handle}h and \ref{fig:failcase-short}h, respectively),
although their vector forms are significantly different from the ground-truth
images (compare Figures~\ref{fig:failcase-intersect}b,
\ref{fig:failcase-anglechange}b, \ref{fig:failcase-handle}b and
\ref{fig:failcase-short}b with \ref{fig:failcase-intersect}d,
\ref{fig:failcase-anglechange}d, \ref{fig:failcase-handle}d and
\ref{fig:failcase-short}d, respectively). This situation results in low data
energy, especially when the resolution of the input image is relatively low.

Our prior energy is designed precisely to solve the above-mentioned problems
and to ensure that the resultant bezigons are reasonable and aesthetically
pleasing. More specifically, we construct a prior functional to reduce the
likelihood of each type of failure cases. Thus, our prior energy has the
following form:
\begin{equation}
\begin{split}
E_{prior} (B) &= \lambda_{spt} E_{spt}  (B) + \lambda_{apt} E_{apt}  (B) \\
               &+ \lambda_{hpt} E_{hpt}  (B) + \lambda_{lpt} E_{lpt}  (B)
\end{split}
\end{equation}
where $E_{spt}$, $E_{apt}$, $E_{hpt}$ and $E_{lpt}$ represent the
self-intersection prior term (SPT), the angle-variation prior term (APT), the
B\'ezier-handle prior term (HPT) and the curve-length prior term (LPT),
respectively, and $\lambda_{spt}$, $\lambda_{apt}$, $\lambda_{hpt}$ and
$\lambda_{lpt}$ are their respective weights. Each of the prior terms is
specifically defined and explained in the following subsections.

\subsubsection{Elimination of Self-intersection}

\begin{figure}[!t]
\centering
    \subfigure[] {\includegraphics[height=0.7in]
                 {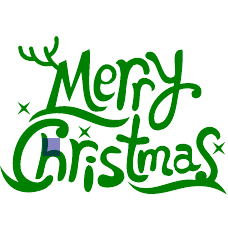}}
    \subfigure[] {\includegraphics[height=0.7in]
                 {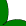}}
    \subfigure[] {\includegraphics[height=0.7in]
                 {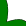}}
    \subfigure[] {\includegraphics[height=0.7in]
                 {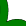}}
    \subfigure[] {\includegraphics[height=0.7in]
                 {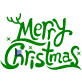}}
    \subfigure[] {\includegraphics[height=0.7in]
                 {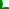}}
    \subfigure[] {\includegraphics[height=0.7in]
                 {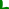}}
    \subfigure[] {\includegraphics[height=0.7in]
                 {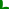}}
\caption{An example of eliminating self-intersection. (a) The entire
ground-truth vector image and the local patch to be processed. (b) Result of
optimization without the SPT. (c) Result of optimization with the SPT. (d)
Ground truth. (e-h) are the rasterization results corresponding to (a-d),
respectively.}
\label{fig:failcase-intersect}
\end{figure}

Certain approaches are seemingly capable of avoiding self-intersection but are
not feasible in practice. One intuitive method is to enforce a set of highly
coupled nonlinear inequality constraints and use a primal-dual interior point
method \cite{mehrotra1992implementation} for optimization. However, this
approach is not suitable in our case because of its computational complexity.
As another na\"{\i}ve method, we could assign a large constant energy to a
bezigon that is detected as exhibiting self-intersection. However, this
provides almost no guidance for a bezigon that has already manifested
self-intersection during optimization.

\begin{figure}[!t]
\centering
    \subfigure[] {\includegraphics[width=0.7in]{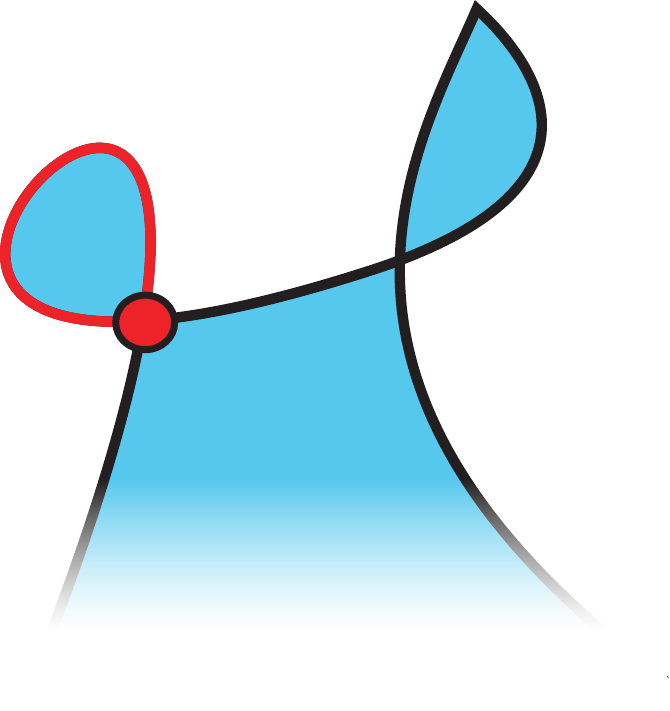}}
    \subfigure[] {\includegraphics[width=0.7in]{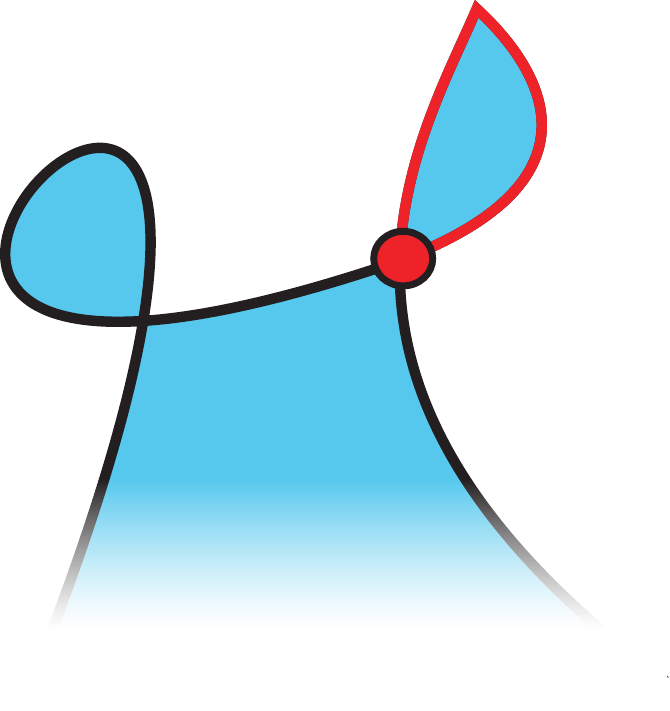}}
\caption{Measuring the extent of self-intersection. (a) and (b) address the
first and second intersection points, respectively. The shorter divided portion
in each phase is indicated by the red curve.}
\label{fig:remove-intersect}
\end{figure}

Instead, we attempt to analytically measure the extent of self-intersection and
provide an effective regularization to automatically avoid bezigons with
self-intersection. The primary advantage of our method is that it not only is
capable of preventing self-intersection but also provides effective guidance to
eliminate self-intersection that has already occurred. Moreover, it does not
require expensive computation.

The procedure is illustrated in Figure~\ref{fig:remove-intersect}. We first
estimate all intersection points (indicated by red dots), if any. Each such
point divides the bezigon outline into two parts. We consider the shorter of
these parts (shown as red curves) and measure the extent of self-intersection
by summing over their lengths. More formally, the measurement can be written as
\begin{equation}
E_{spt} (B) = \sum_{(t_1,t_2) \in T} \min⁡ (L(t_1,t_2), L(0,N)-L(t_1,t_2)).
\end{equation}
Here, $T = \{(t_1,t_2)|t_1<t_2, S(t_1)=S(t_2)\}$ is the set of partitions
corresponding to all intersection points (red dots in
Figure~\ref{fig:remove-intersect}), and $L(t_1,t_2)$ represents the arc length
along the curve $S$ from $t_1$ to $t_2$ , i.e.,
\begin{equation}
L(t_1,t_2) \! = \!\!\!\!\!\! \sum_{j=\lfloor t_1 \rfloor
+ 1}^{\lceil t_2 \rceil} \!\! \int_{t_{j,1}}^{t_{j,2}} \!\!
\sqrt{\left(\frac{dX_j(B_x,t)}{dt}\right)^2 \!\!+\!\!
\left(\frac{dY_j(B_y,t)}{dt}\right)^2} \!\! dt,
\end{equation}
where $t_{j,1} = \max(t_1-j+1, 0)$ and $t_{j,2} = \min(t_2-j+1,1)$.

The energy term $E_{spt}$ penalizes significant self-intersection. The more
severe an intersection is, the more closely the length of a shorter part
approaches the length of a longer part, and hence, the larger $E_{spt}$ will
be. When there is no self-intersection, $E_{spt}$ is equal to zero. Our
experiments demonstrate that optimization using the SPT results in bezigons
that contain very little self-intersection and are likely to be close to the
ground-truth image in terms of topology (see
Figure~\ref{fig:failcase-intersect}c).

\subsubsection{Regularization for Angle Variations}

\begin{figure}[!t]
\centering
    \subfigure[] {\includegraphics[height=0.7in]
                 {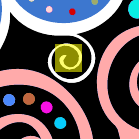}}
    \subfigure[] {\includegraphics[height=0.7in]
                 {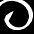}}
    \subfigure[] {\includegraphics[height=0.7in]
                 {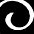}}
    \subfigure[] {\includegraphics[height=0.7in]
                 {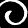}}
    \subfigure[] {\includegraphics[height=0.7in]
                 {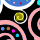}}
    \subfigure[] {\includegraphics[height=0.7in]
                 {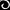}}
    \subfigure[] {\includegraphics[height=0.7in]
                 {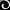}}
    \subfigure[] {\includegraphics[height=0.7in]
                 {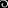}}
\caption{An example of regularization for angle variations. (a) The entire
ground-truth vector image and the local patch to be processed. (b) Result of
optimization without the APT. (c) Result of optimization with the APT. (d)
Ground truth. (e-h) are the rasterization results of (a-d), respectively.}
\label{fig:failcase-anglechange}
\end{figure}

Although a simple curve-smoothing algorithm may remove small angle variations,
such a method will most likely fail to preserve other visually significant
corners. Moreover, it may not always be possible to identify the saliency of
the corners using a fixed threshold for angle variations. Consequently, we must
develop a more sophisticated method of smoothing out insignificant corners
while preserving the small number of significant corners.

For this purpose, we penalize the sum of all angle variations. As a result, the
optimized bezigon will consist of predominantly zero-angle variations and a
small number of non-zero angle variations. This is important because it
incorporates corner detection into the bezigon optimization.

More formally, we denote the two tangent vectors of the $j$-th endpoint by
$\mathbf{a}_j=(x_{j,1} - x_{j-1,3},y_{j,1} - y_{j-1,3})$ and
$\mathbf{b}_j=(x_{j,2} - x_{j,1},y_{j,2} - y_{j,1})$. Then, the APT can be
written as follows:
\begin{equation}
E_{apt}(B) = \sum_{j=1}^N \cos^{-1} \frac{\mathbf{a}_j \cdot
\mathbf{b}_j}{\|\mathbf{a}_j\|\|\mathbf{b}_j\|}.
\end{equation}

The experimental results demonstrate that optimization with the APT can retain
the smoothness of the bezigon while preserving visually significant corners
(Figure~\ref{fig:failcase-anglechange}c).

\subsubsection{Avoidance of Short B\'ezier Handles}

\begin{figure}[!t]
\centering
    \subfigure[] {\includegraphics[height=0.7in]
                 {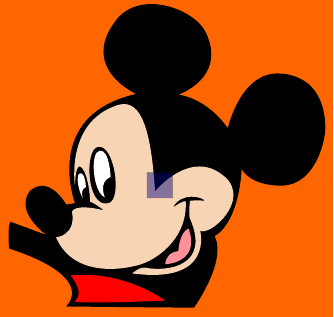}}
    \subfigure[] {\includegraphics[height=0.7in]
                 {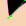}}
    \subfigure[] {\includegraphics[height=0.7in]
                 {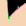}}
    \subfigure[] {\includegraphics[height=0.7in]
                 {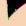}}
    \subfigure[] {\includegraphics[height=0.7in]
                 {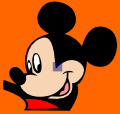}}
    \subfigure[] {\includegraphics[height=0.7in]
                 {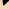}}
    \subfigure[] {\includegraphics[height=0.7in]
                 {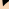}}
    \subfigure[] {\includegraphics[height=0.7in]
                 {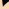}}
\caption{An example of the avoidance of short B\'ezier handles. (a) The entire
ground-truth vector image and the local patch to be processed. (b) Result of
optimization without the HPT. (c) Result of optimization with the HPT. (d)
Ground truth. (e-h) are the rasterization results of (a-d), respectively.}
\label{fig:failcase-handle}
\end{figure}

To guard against the possibility of short handles, we penalize any short handle
using an inverse barrier function
\begin{equation}
E_{hpt}(B) = \sum_{j=1}^N
\left(\begin{split}
&\frac{1}{\sqrt{(x_{j,2} \!-\! x_{j,1})^2 + (y_{j,2} \!-\! y_{j,1})^2}} \\
          &+ \frac{1}{\sqrt{(x_{j\!+\!1,1} \!-\! x_{j,3})^2 + (y_{j\!+\!1,1}
          \!-\! y_{j,3})^2}}
\end{split}\right).
\end{equation}

This term imposes a large penalty on short handles because $E_{hpt}(B)$ tends
toward infinity as any handle length tends toward $0$. When the length of each
handle is sufficiently large, this energy will be very small and hence will not
engender serious side effects.

It can be experimentally demonstrated that optimization using the HPT
considerably reduces the occurrence of unnatural bezigons, as illustrated in
Figure~\ref{fig:failcase-handle}b. Although the resultant handle may sometimes
be slightly longer than it should be (compare the locations of the control
points in Figure~\ref{fig:failcase-handle}c with those in
Figure~\ref{fig:failcase-handle}d), such a result only affects the quality or
editability of the vector results in general.

\subsubsection{Curve-Length Prior}

\begin{figure}[!t]
\centering
    \subfigure[] {\includegraphics[height=0.7in]
                 {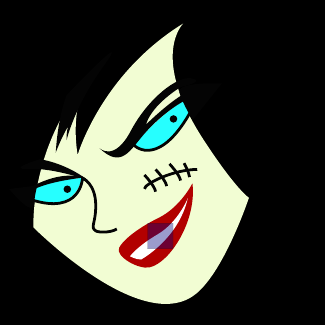}}
    \subfigure[] {\includegraphics[height=0.7in]
                 {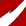}}
    \subfigure[] {\includegraphics[height=0.7in]
                 {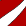}}
    \subfigure[] {\includegraphics[height=0.7in]
                 {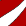}}
    \subfigure[] {\includegraphics[height=0.7in]
                 {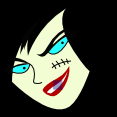}}
    \subfigure[] {\includegraphics[height=0.7in]
                 {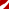}}
    \subfigure[] {\includegraphics[height=0.7in]
                 {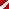}}
    \subfigure[] {\includegraphics[height=0.7in]
                 {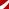}}
\caption{An example of the application of the curve-length prior. (a) The
entire ground-truth vector image and the local patch to be processed. (b)
Result of optimization without the LPT. (c) Result of optimization with the
LPT. (d) Ground truth. (e-h) are the rasterization results of (a-d),
respectively.}
\label{fig:failcase-short}
\end{figure}

Based on the experience of experts who are proficient in image tracing, a
traced curve tends to be stretched as much as possible unless there is strong
evidence that the curve should shrink or twist. This prior knowledge can be
used to eliminate invalid bezigons of this type because the occurrence of a
highly twisted curve without strong evidence in support of such twisting from
the raster input can be easily identified.

Based on such prior knowledge, we penalize the curve length to avoid invalid
bezigons of this type. Thus, the LPT can be defined as follows:
\begin{equation}
E_{lpt}(B) = \sum_{j=1}^N \int_{0}^{1} \sqrt{\left(\frac{dX_j(t)}{dt}\right)^2
           + \left(\frac{dY_j(t)}{dt}\right)^2} dt.
\end{equation}

Figure~\ref{fig:failcase-short}c demonstrates that the unexpectedly twisted
curve of Figure~\ref{fig:failcase-short}b, for which there is insufficient
evidence in the observed data, can be avoided through the adoption of the LPT.

\subsection{Piecewise Bezigon Optimization}
\label{subsec:piecewise}

Once we have obtained the energy function developed in the previous sections,
in many cases, this function can be minimized using a general non-linear
optimization method. However, because a typical bezigon often consists of a
large number of parameters to be optimized and because the valid range for each
parameter is large, the efficiency and even the convergence of the optimization
process might be an issue. However, bezigon parameters possess a strong local
control property. We may reduce the number of redundant calculations by fully
utilizing this property.

In this section, we explore the nature of bezigon parameters and propose a
piecewise optimization strategy that allows our high-dimensional problem to be
decomposed into several subcomponents that may be individually solved.

The fundamental concept of piecewise optimization is to optimize only a subset
of the geometric parameters of each bezigon at any given time. This task is
feasible because the effect of varying any given control point is limited to a
local region of the bezigon.

More specifically, we regard two consecutive B\'ezier curve sections as one
curve piece. Therefore, a bezigon with $N$ curve sections also consists of $N$
overlapped curve pieces. All curve pieces will be successively optimized. When
optimizing a curve piece, we fix the first and last endpoints of the curve
piece and determine the optimal solution for the four intermediate control
points and the middle endpoint. We first optimize the five active control
points (the red points in Figure~\ref{fig:piecewise}a) of one curve piece and
subsequently optimize the corresponding points (the red points in
Figure~\ref{fig:piecewise}b) of the next curve piece. It should be noted that
the two consecutive pieces overlap and that two of the intermediate control
points (e.g., those shown in red in both Figure~\ref{fig:piecewise}a and
Figure~\ref{fig:piecewise}b) are shared. Therefore, all intermediate control
points will be optimized twice in individual iterations. The process
iteratively progresses from the first curve piece to the last.

\begin{figure}[!t]
\centering
    \subfigure[] {\includegraphics[width=1.5in]{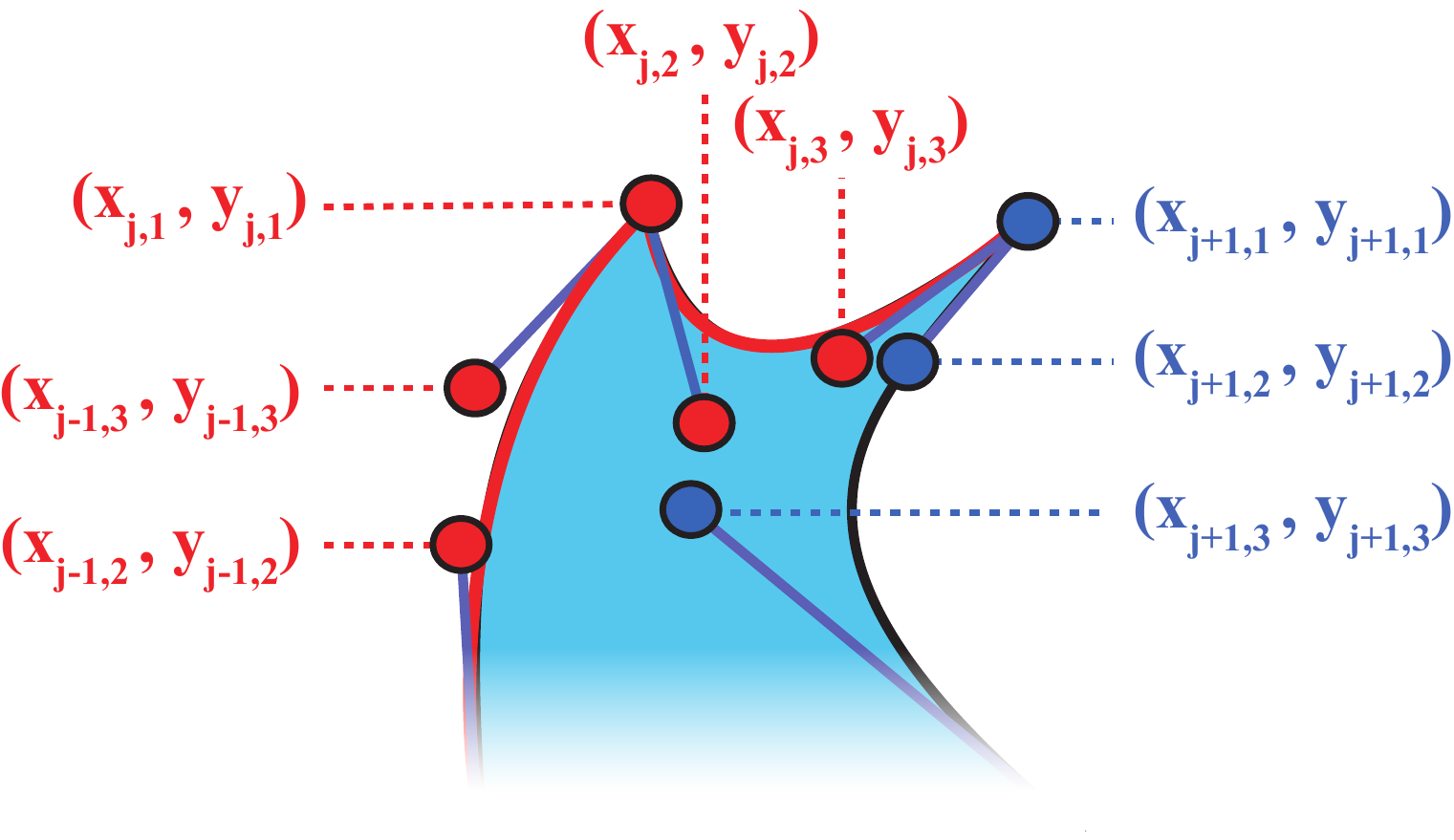}}
    \subfigure[] {\includegraphics[width=1.5in]{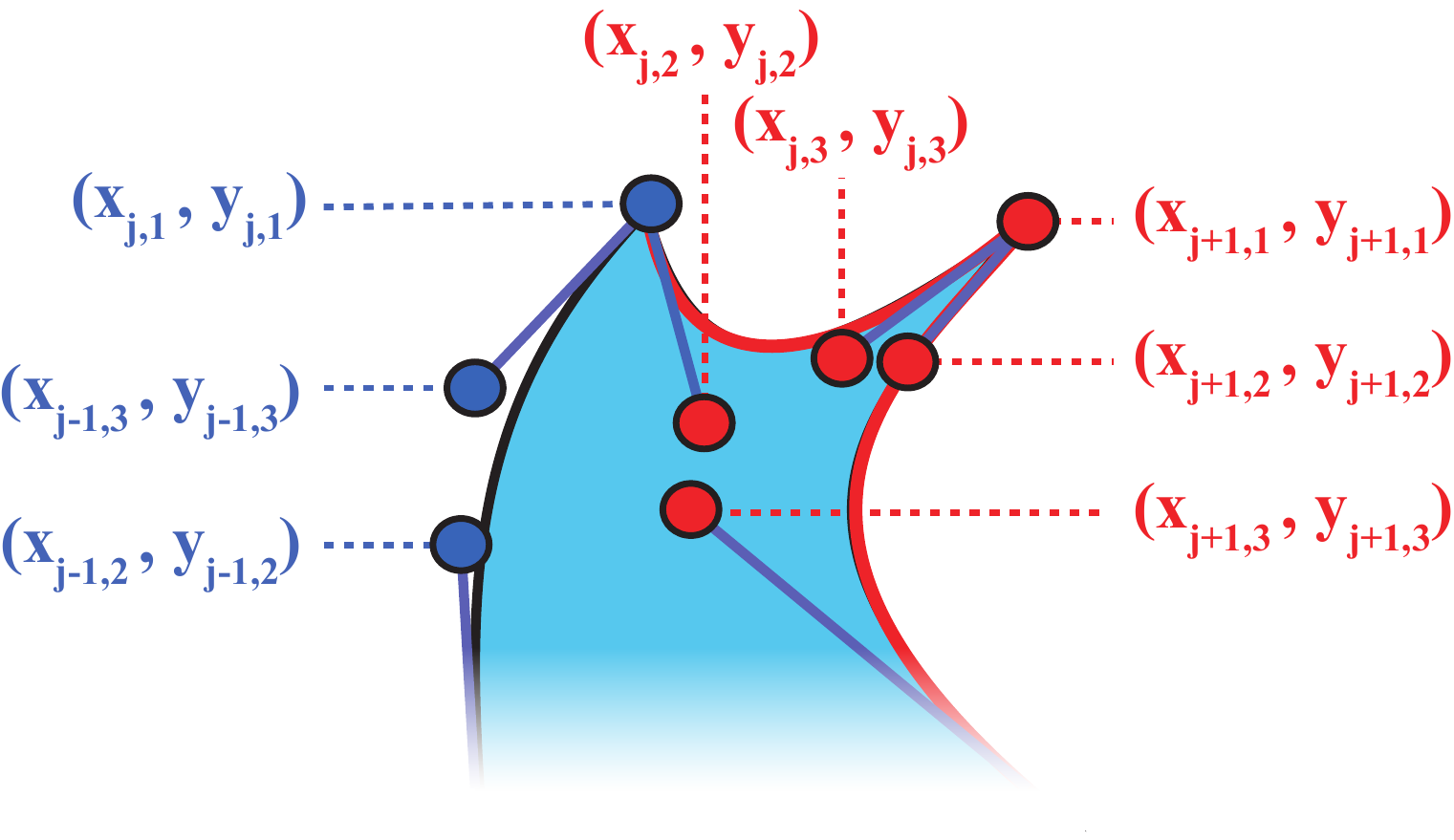}}
\caption{Overlapped piecewise optimization. The curve sections being optimized
are shown in red. (a) Optimizing the curve piece $P_{j-1}$. (b) Optimizing the
curve piece $P_j$.}
\label{fig:piecewise}
\end{figure}

Formally, we represent the geometric parameters of the $j$-th piece to be
optimized as follows:
\begin{equation}
P_j = \begin{pmatrix}
x_{j,2} & x_{j,3} & x_{j+1,1} & x_{j+1,2} & x_{j+1,3} \\
y_{j,2} & y_{j,3} & y_{j+1,1} & y_{j+1,2} & y_{j+1,3}
\end{pmatrix}.
\end{equation}

All remaining geometric parameters $\tilde{P_j}$ of the bezigon are held fixed
during the present optimization. Therefore, optimizing this curve piece amounts
to identifying the optimal configuration $P_j^*$ that minimizes a function
composed of the local energies, i.e.,
\begin{equation}
\label{eq:P_j*}
P_j^* = \arg \min_{P_j \in \Omega_j} \left[ E_{data} (P_j; \tilde{P_j}, I) +
E_{prior} (P_j; \tilde{P_j}) \right].
\end{equation}
Here, $\Omega_j$ is the space of all possible geometric parameters for the
$j$-th piece.

Such a strategy substantially increases the efficiency of the entire
optimization. Although the objective function~\ref{eq:P_j*} is quite similar to
Equation~\ref{eq:w*}, the solution space $\Omega_j=\mathbb{R}^{10}$ is much
smaller than $\Omega$. Therefore, the original high-dimensional problem can be
decomposed into a set of lower-dimensional problems, which greatly improves the
efficiency of the overall optimization process. Moreover, all prior energy
terms, except the SPT, are simply the sum of the corresponding energy of each
curve section. Therefore, we can consider only two related sections when
calculating these terms. In this manner, a large number of redundant
computations can be eliminated.

Piecewise optimization not only is fast but also provides satisfactory bezigons
with almost no decrease in accuracy. The experimental results indicate that
after all curve pieces are traversed two or three times, in most cases, the
resultant bezigon is nearly perfect.

As an optional step, we can jointly optimize all geometric parameters once more
to further improve the result. Because our piecewise procedure can provide
substantially more accurate input for further optimization, this subsequent
global optimization can be significantly more efficient than it would be
without piecewise optimization.  The entire process of bezigon optimization is
summarized in Algorithm~\ref{alg:bezigon-optimization}.

\begin{algorithm}[h!]
    \caption{Bezigon Optimization}\label{alg:bezigon-optimization}
    \begin{algorithmic}[1]
        \REPEAT {
            \FOR{$j=1$ \TO $N$} {
                \STATE Optimize $P_j$ according to Equation~\ref{eq:P_j*}
                \STATE Update $B$ according to $P_j$
            }
            \ENDFOR
        }
        \UNTIL{converged}
        \STATE Optimize $B$ according to Equation~\ref{eq:w*} (optional)
    \end{algorithmic}
\end{algorithm}

The overlapped piecewise optimization strategy provides a fast yet accurate
method of bezigon optimization, which is an important prerequisite for the
practical application of our vectorization approach.

\section{Experiments}
\label{sec:experiments}

To demonstrate the effectiveness of the approach developed in this paper, in
this section, we quantitatively and qualitatively compare our method with other
vectorization methods. As stated in Section~\ref{sec:related-work}, many
vectorization algorithms and software packages exist. However, most academic
work on vectorization is not relevant for comparisons because it is primarily
focused on photographs or other types of vectorization. We restrict our
comparisons to two approaches that are specialized for the vectorization of
clipart images. One method is Vector Magic~\cite{vm2010}, which was developed
on the basis of the state-of-the-art method proposed by
\cite{diebel2008bayesian} \footnote{\cite{zhang2009vectorizing} may produce a
more reasonable segmentation than \cite{diebel2008bayesian}, especially when
the input image contains non-uniform color regions or decorative lines, and is
probably the best method for vectorizing clipart animations. However, as
previously stated, neither segmentation nor video vectorization is one of our
topics of focus. Therefore, we do not offer a comparison with
\cite{zhang2009vectorizing} in the present work.}. The other software package
we consider is Adobe Illustrator \cite{adobe2013}, which is a representative
example of a widely used commercial vectorization software. The experimental
results demonstrate that our approach is superior to these methods in terms of
bezigon quality.

\subsection{Implementation}

To evaluate the effectiveness of the proposed bezigon optimization method, we
have implemented a prototype image vectorization system.

The core of our system is bezigon optimization. Because of the continuity and
differentiability of our energy function, either a curve piece or a global
bezigon can be effectively optimized using many available optimization
algorithms (such as NEWUOA~\cite{powell2006newuoa},
l-BFGS~\cite{liu1989limited}, and the conjugate gradient
method~\cite{hestenes1952methods}).  Note that there are four tuning parameters
in our objective function, namely, the weights of the four prior terms.
Empirically, we set $\lambda_{spt}=1.0$ (to strongly penalize
self-intersection), $\lambda_{apt}=0.08$, $\lambda_{hpt}=0.1$ and
$\lambda_{lpt}=0.1$. Although there may be other weight settings that would
yield better performance, we did not perform a thorough search for the optimal
weights. According to our experimental results, the proposed method is
generally insensitive to these parameters. Our preset weights should yield
satisfactory results. However, as the quality of the raster input decreases
(e.g., low-resolution input), fine tuning $\lambda_{apt}$ may become necessary
to generate perfect bezigons. Although our framework does not intrinsically
rely on any assumption regarding the color model, for simplicity of
implementation and convenience of fair comparisons with the most commonly used
clipart image vectorization methods, our prototype system currently assumes
that the color in each bezigon is uniform \footnote{To the best of our
knowledge, this assumption is also adopted by most existing clipart image
vectorization approaches, including Vector Magic and Adobe Illustrator.}, i.e.,
$c(C;x,y)=C_0$ and the color parameter $C=C_0$ is an arbitrary vector in the
RGB color space.

The initial bezigons can be either extracted from the input image or obtained
using other vectorization methods. They are not required to be highly accurate.
Most initial bezigons in our experiments are far from perfect. Of course, if
the initial pose drifts too far from the optimal pose (e.g., approaches random
bezigons), our bezigon optimization may become trapped in a local optimum and
output imperfect results. However, in practice, this rarely occurs because it
is not very difficult to estimate a bezigon that is sufficient to serve as an
initial solution. The real difficulties lie in the subsequent optimization,
i.e., achieving bezigons of even higher precision, which is the key issue
addressed in this paper.

Note that because this prototype system was primarily developed as a proof of
concept, the speed of the process is not a priority at the moment. Our
implementation code is currently written in Python, a dynamically typed and
interpreted language. The code is run on a laptop with an Intel Core i5-2410M @
2.53 GHz processor with 4 GB of memory. The total execution time varies (10
secs to 10 mins) as a function of the complexity of the shapes to be
vectorized. It should be much faster when implemented in a static language.
Moreover, our method can be highly parallelized by virtue of the nature of
wavelet rasterization, which may also considerably improve the efficiency.

\subsection{Quantitative Comparisons}

We use a fidelity metric to quantitatively compare our results
with those of other methods. The fidelity metric generally provides a good
indication of the characteristics that define a good vectorization algorithm.
To further evaluate the proposed method based on human aesthetic judgment, we
also present a user study.

For both comparisons, we collected a set of clipart images available in both
raster and vector formats. All the raster images served as inputs to our
algorithm and to the other vectorization methods. Some methods considered in
the comparison require parameter tuning. To perform a fair comparison, the
dominant parameters of these methods were adjusted until the number of bezigons
produced as output were approximately equal to the number of bezigons in the
ground-truth vector.  Then, we compared the vector images resulting from the
different methods with respect to fidelity and user satisfaction. The details
of both comparisons are presented below.

\textbf{Comparison via peak signal-to-noise ratio measurement.} The quality of a
vectorization is often evaluated in terms of the PSNR (peak signal-to-noise
ratio) or RMSE (root-mean-square error)
\cite{diebel2008bayesian,zhang2009vectorizing}. Before evaluation,
both the resultant vector image and the ground-truth image were rasterized at a
specific resolution.

Figure~\ref{fig:PSNR} presents the histograms of the the increases in the PSNR
that were achieved by our method with respect to Vector Magic
(Figure~\ref{fig:PSNR}a) and Adobe Illustrator (Figure~\ref{fig:PSNR}b). It is
evident that our method consistently yields a higher PSNR compared with
competing methods. More specifically, our results reveal an increase in the
PSNR of 0-5 dB with respect to Vector Magic and an increase of 10-20 dB with
respect to Adobe Illustrator.

\begin{figure}[!t]
\centering
    \subfigure[] {\includegraphics[width=1.6in]{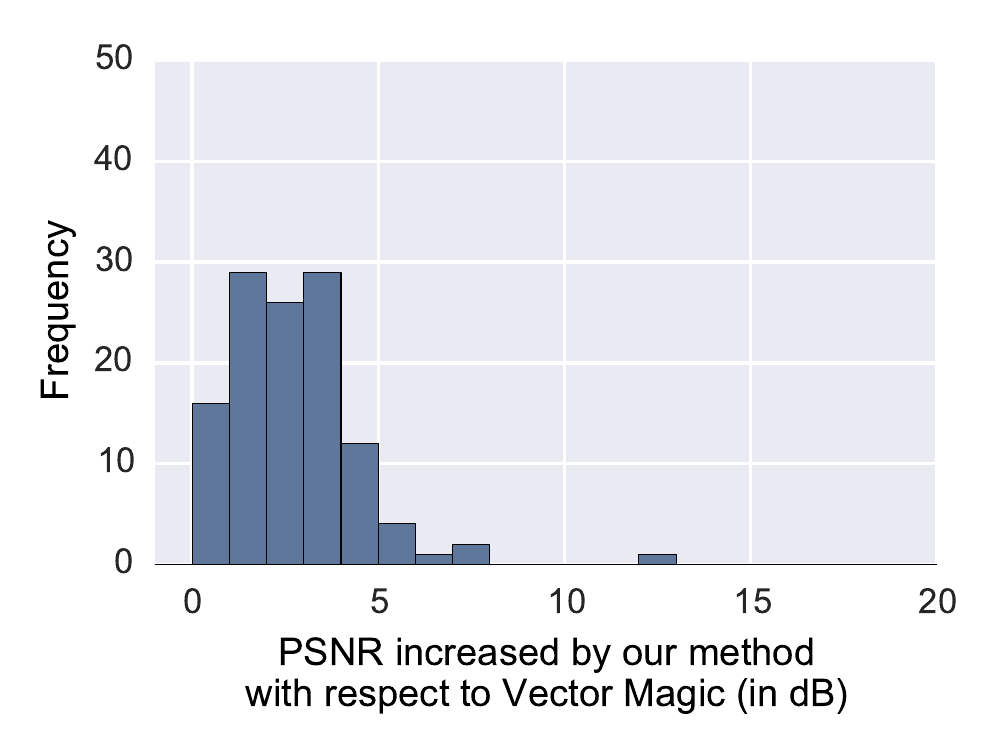}}
    \subfigure[] {\includegraphics[width=1.6in]{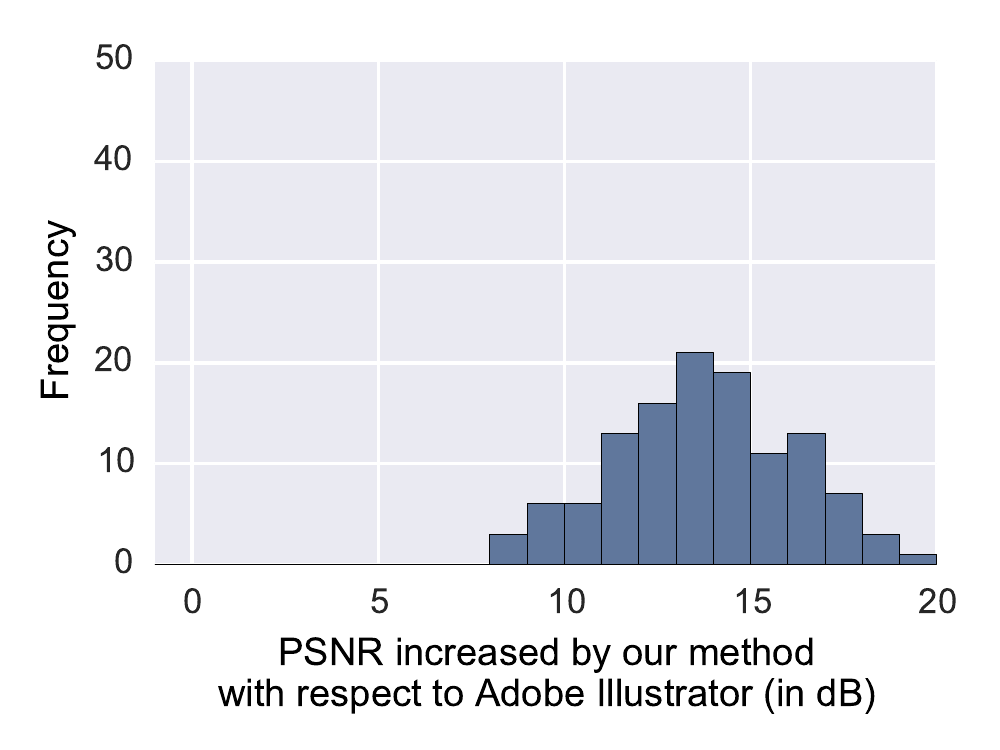}}
\caption{Histograms of the increases in the PSNR achieved by our method
compared with (a) Vector Magic and (b) Adobe Illustrator.}
\label{fig:PSNR}
\end{figure}

\textbf{Comparison via a user study.} We also performed a user study to obtain
a further evaluation based on human aesthetic judgment. For this purpose, a
pairwise comparison test was created. We prepared 120 pairs of vector results.
Each pair consisted of a vector image generated by our method and a vector
image generated by another method (either Vector Magic or Adobe Illustrator).
We constructed a web interface to show each pair of vector images, including
their control points, but with no creator vectorizer name attached. Several
participants with graphic design backgrounds were then asked to determine
whether one image was much better than, better than, almost the same as, a
little worse than, or much worse than another image in comparison with the
ground truth. The statistical results are presented in
Figure~\ref{fig:userstudy}. This figure indicates that our results were
considered to be superior those of the current state-of-the-art method (Vector
Magic) in nearly 80\% of the pairwise comparisons. Approximately one quarter of
the images were deemed to be much better (Figure~\ref{fig:userstudy}a).
Compared with the representative commercial software (Adobe Illustrator),
almost all of our results were considered to be better, and half of them were
judged to be much better.

\begin{figure}[!t]
\centering
    \subfigure[] {\includegraphics[width=1.6in]{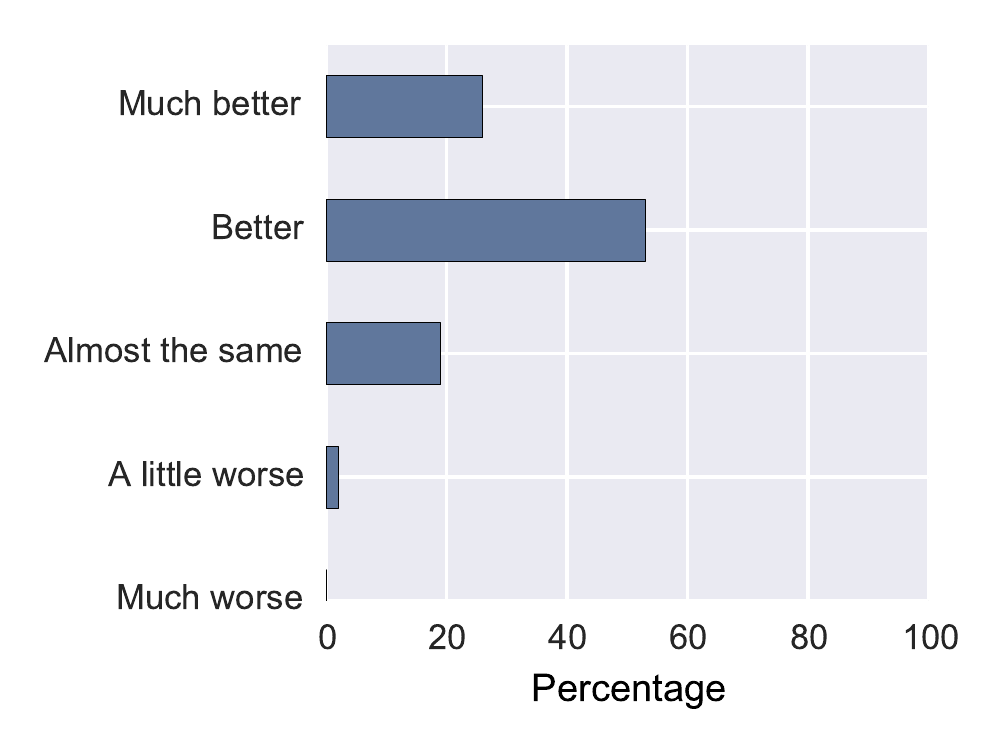}}
    \subfigure[] {\includegraphics[width=1.6in]{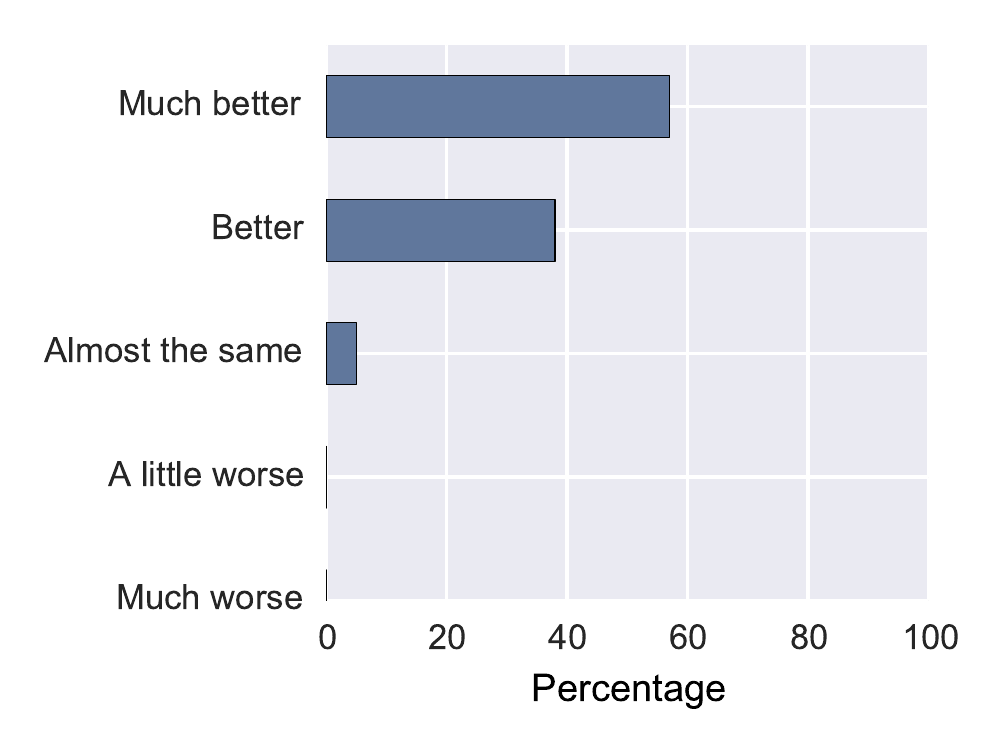}}
\caption{Results from a user study that compared our results with those of (a)
Vector Magic and (b) Adobe Illustrator.}
\label{fig:userstudy}
\end{figure}

To summarize the quantitative comparisons, our approach is found to be superior
to both the state-of-the-art algorithm and the representative commercial tool
in terms of both fidelity and user satisfaction.

\subsection{Qualitative Comparisons}

For qualitative comparison of our method with the other methods, we provide a
few results
(Figures~\ref{fig:visual-comparison-smooth}-\ref{fig:visual-comparison-bent})
obtained using our approach and the competing methods. Because of space
limitations, we highlight only one local patch for each image (shown in the
even rows of
Figures~\ref{fig:visual-comparison-smooth}-\ref{fig:visual-comparison-bent}).
From the comparison, we observe that our results, in general, are more faithful
to the raster input and that the shapes of the resultant bezigons are more
reasonable and visually pleasing. More specifically, our strengths lie in the
following cases.

\textbf{Case 1: smooth boundary with high curvature.} In a typical
clipart image, smooth boundaries with high curvature are often found in the
round corners of a shape (Figure~\ref{fig:visual-comparison-smooth}j and
Figure~\ref{fig:visual-comparison-smooth}t). Traditional methods typically use
a chain of densely sampled points to represent such structures and subsequently
fit B\'ezier curves to the chain. The problem with this approach is that
reconstruction of curves from such an intermediate representation can be
excessively ambiguous in such regions. Therefore, the resultant bezigons
exhibit false corners (see the redundant corners in
Figures~\ref{fig:visual-comparison-smooth}g,
\ref{fig:visual-comparison-smooth}h, and
\ref{fig:visual-comparison-smooth}r). Instead, our method directly infers
the bezigons from the raster input to reduce such ambiguities to the greatest
possible extent. Consequently, our bezigons contain fewer false sharp corners
(Figure~\ref{fig:visual-comparison-smooth}i and
\ref{fig:visual-comparison-smooth}s).

\begin{figure}[!t]
\centering
    \subfigure[] {\includegraphics[width=0.6in]
                 {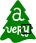}}
    \subfigure[] {\includegraphics[width=0.6in]
                 {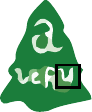}}
    \subfigure[] {\includegraphics[width=0.6in]
                 {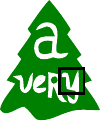}}
    \subfigure[] {\includegraphics[width=0.6in]
                 {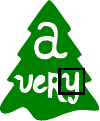}}
    \subfigure[] {\includegraphics[width=0.6in]
                 {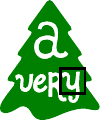}}

    \subfigure[] {\includegraphics[width=0.6in]
                 {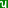}}
    \subfigure[] {\includegraphics[width=0.6in]
                 {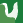}}
    \subfigure[] {\includegraphics[width=0.6in]
                 {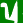}}
    \subfigure[] {\includegraphics[width=0.6in]
                 {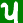}}
    \subfigure[] {\includegraphics[width=0.6in]
                 {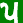}}

    \subfigure[] {\includegraphics[width=0.6in]
                 {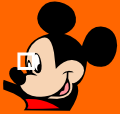}}
    \subfigure[] {\includegraphics[width=0.6in]
                 {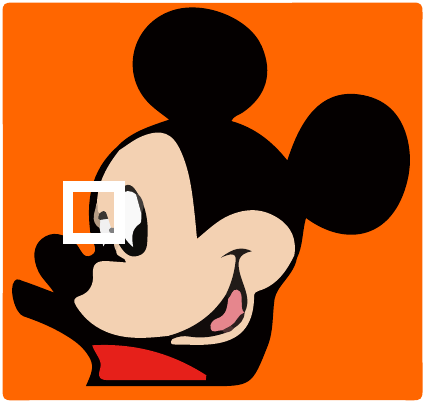}}
    \subfigure[] {\includegraphics[width=0.6in]
                 {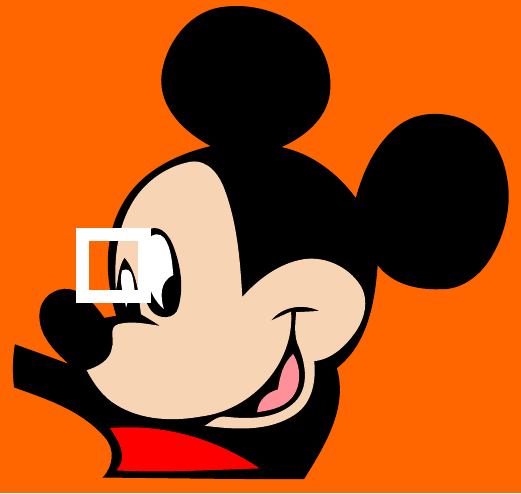}}
    \subfigure[] {\includegraphics[width=0.6in]
                 {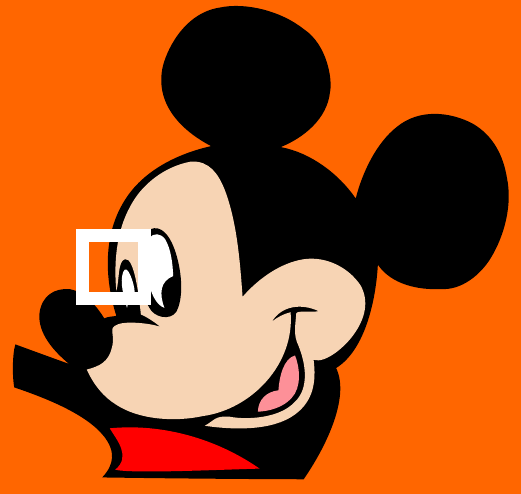}}
    \subfigure[] {\includegraphics[width=0.6in]
                 {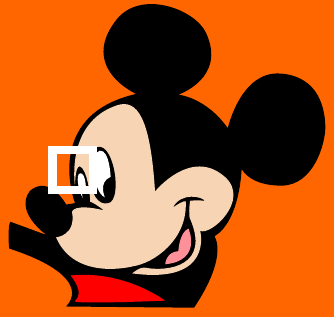}}

    \subfigure[] {\includegraphics[width=0.6in]
                 {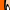}}
    \subfigure[] {\includegraphics[width=0.6in]
                 {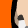}}
    \subfigure[] {\includegraphics[width=0.6in]
                 {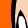}}
    \subfigure[] {\includegraphics[width=0.6in]
                 {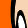}}
    \subfigure[] {\includegraphics[width=0.6in]
                 {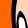}}
\caption{Comparisons in cases of smooth boundaries with high curvature. From
left to right: raster input, Adobe's result, Vector Magic's result, our result,
ground truth.}
\label{fig:visual-comparison-smooth}
\end{figure}

\textbf{Case 2: obtuse corners.} Many shapes to be vectorized contain obtuse
corners (Figures~\ref{fig:visual-comparison-corner}j and
\ref{fig:visual-comparison-corner}t). Preserving such corners is very
important even when the result is visually satisfactory because it can be
difficult to subsequently edit the vectorized shapes. However, for the same
reason discussed in Case 1, traditional methods tend to smooth out such corners
or yield a curve endpoint with an incorrect location (e.g., the overly smoothed
boundaries in Figures~\ref{fig:visual-comparison-corner}h,
\ref{fig:visual-comparison-corner}q and
\ref{fig:visual-comparison-corner}r). Our method benefits from the
direct optimization of the bezigons and avoids errors introduced by fitting
sampled points that cannot accurately indicate the correct location of an
obtuse corner. Therefore, our results typically preserve more obtuse corners
(Figures~\ref{fig:visual-comparison-corner}i and
\ref{fig:visual-comparison-corner}s).

\begin{figure}[!t]
\centering

    \subfigure[] {\includegraphics[width=0.6in]
                 {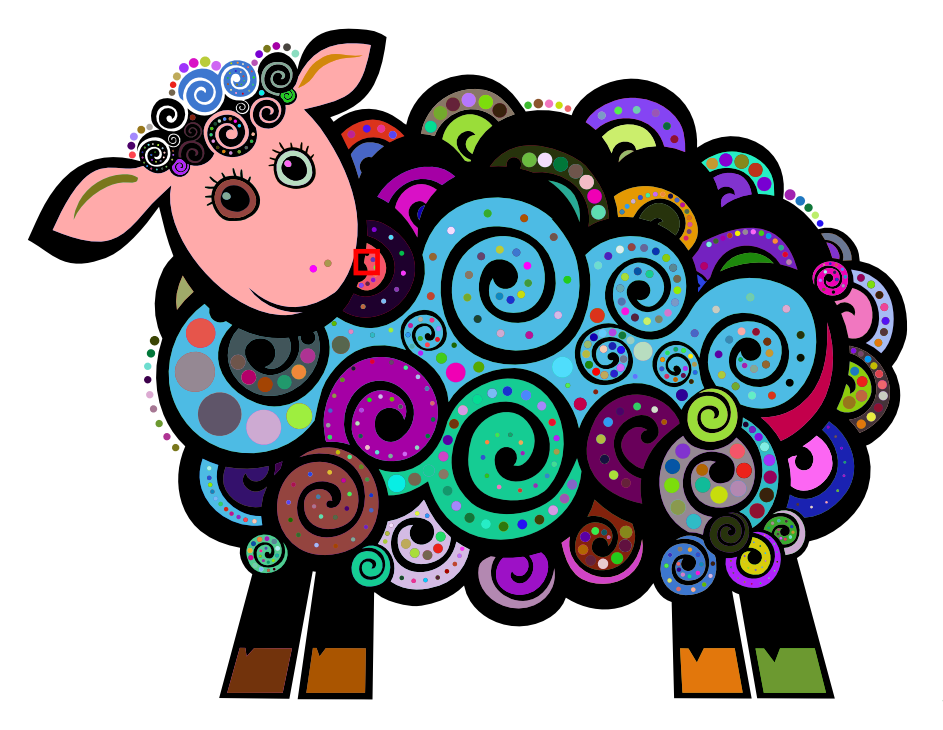}}
    \subfigure[] {\includegraphics[width=0.6in]
                 {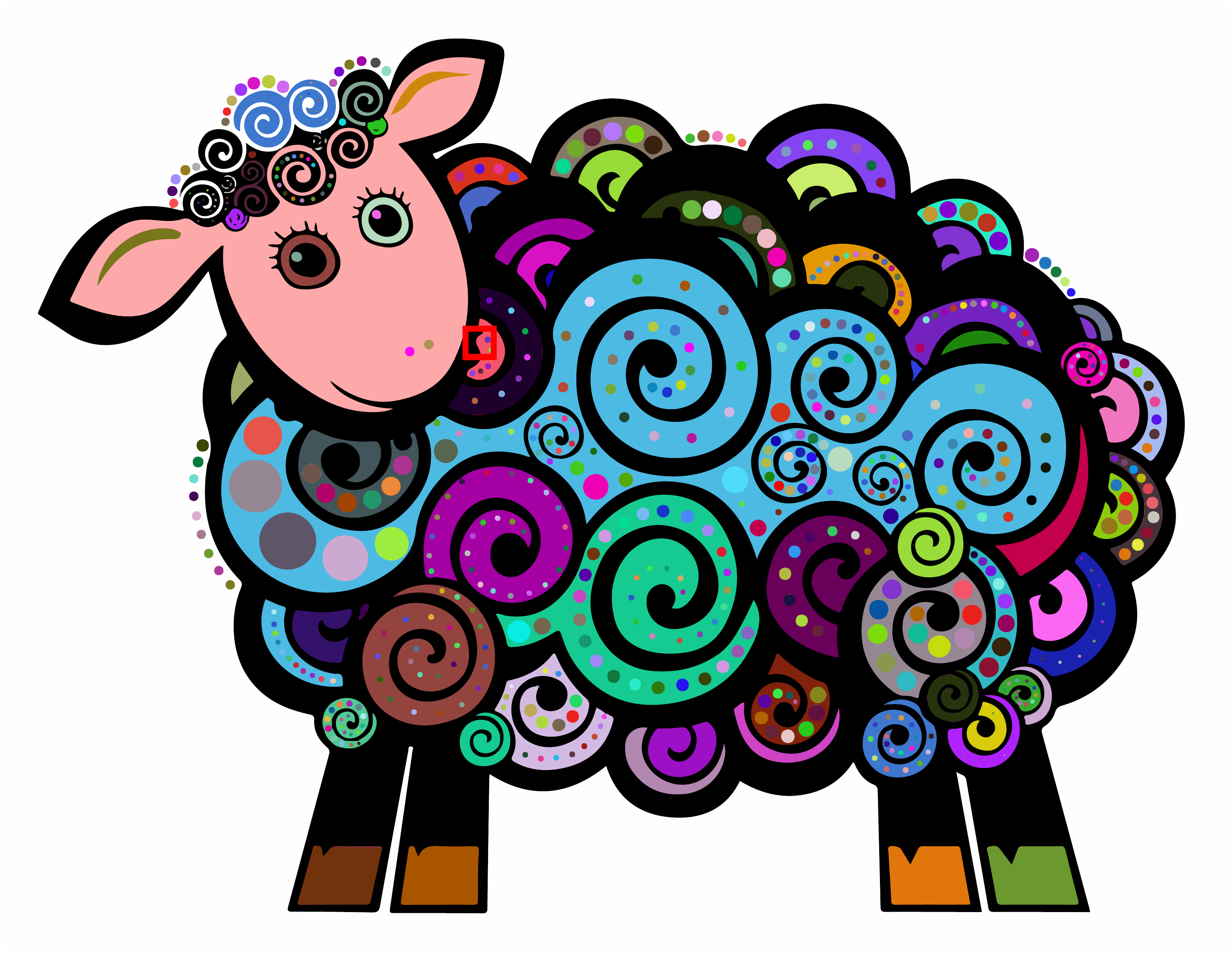}}
    \subfigure[] {\includegraphics[width=0.6in]
                 {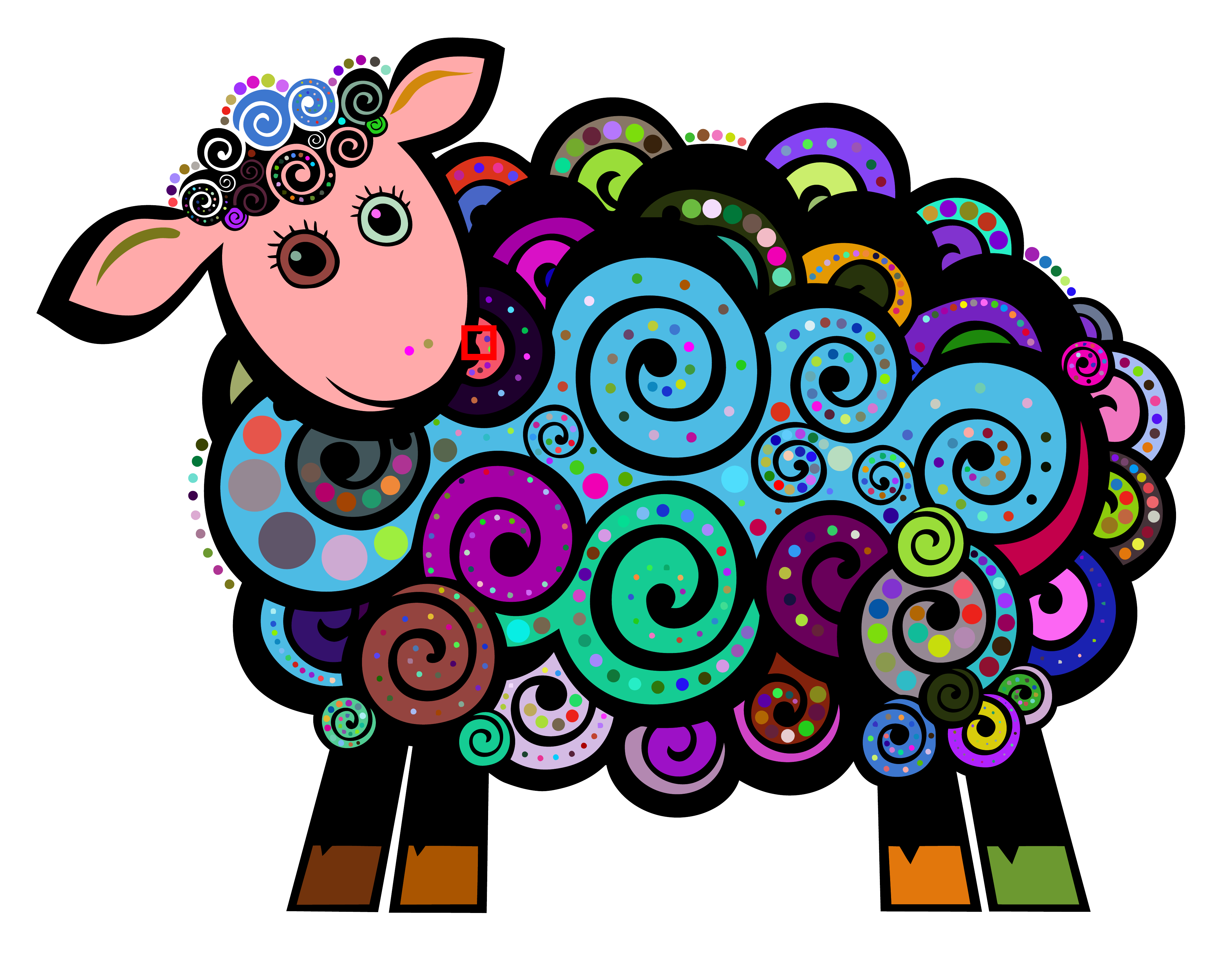}}
    \subfigure[] {\includegraphics[width=0.6in]
                 {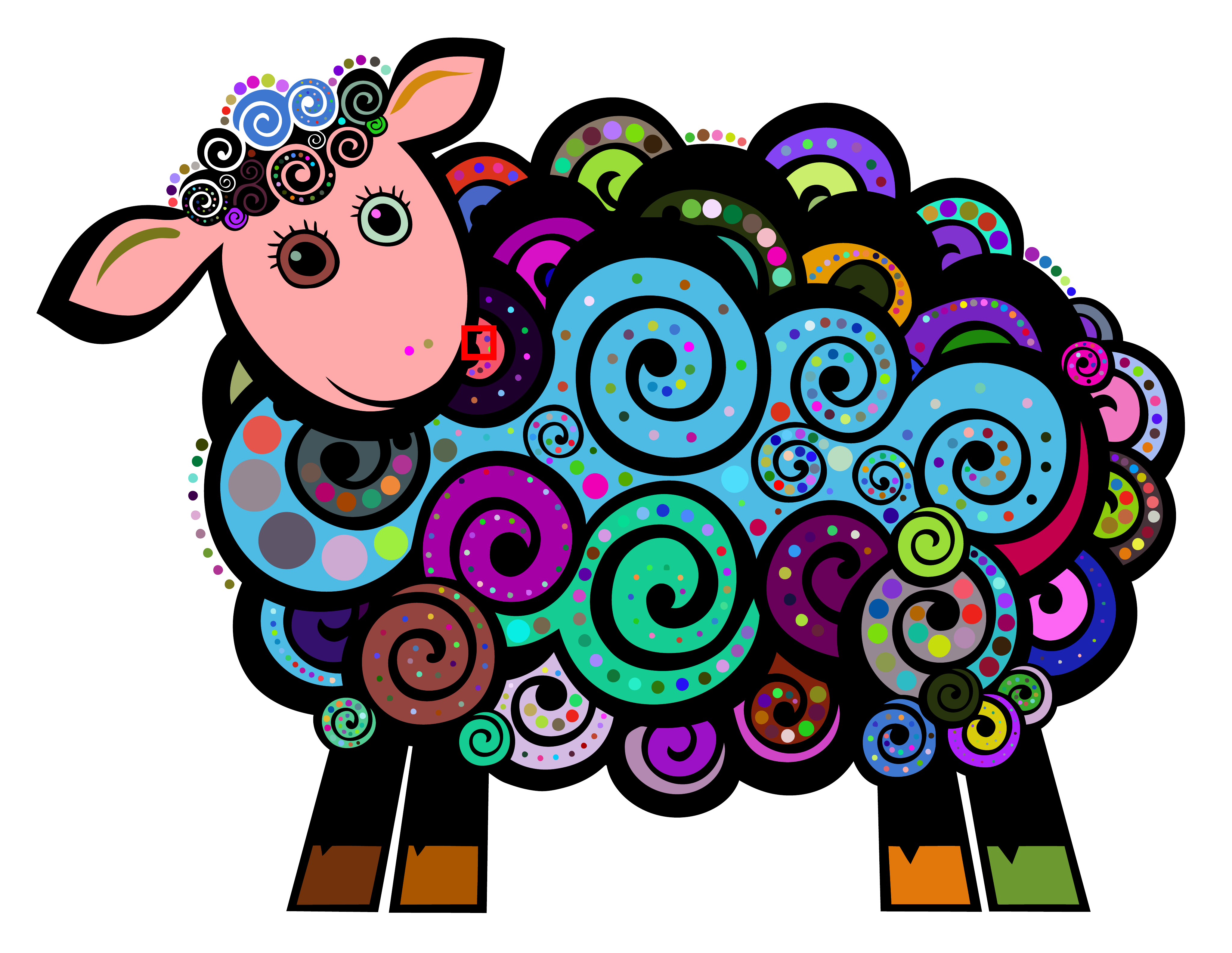}}
    \subfigure[] {\includegraphics[width=0.6in]
                 {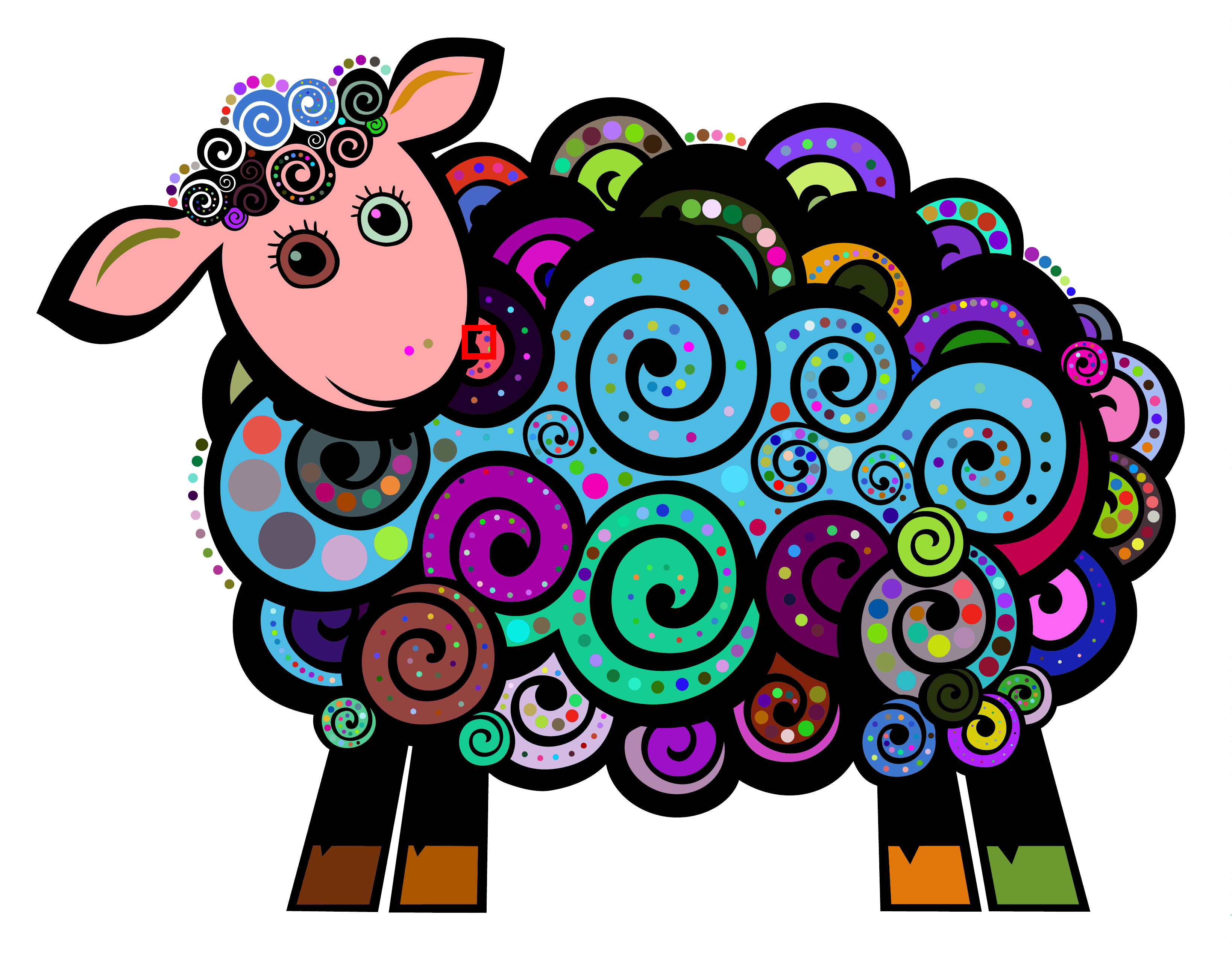}}

    \subfigure[] {\includegraphics[width=0.6in]
                 {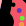}}
    \subfigure[] {\includegraphics[width=0.6in]
                 {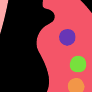}}
    \subfigure[] {\includegraphics[width=0.6in]
                 {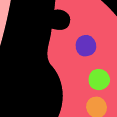}}
    \subfigure[] {\includegraphics[width=0.6in]
                 {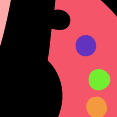}}
    \subfigure[] {\includegraphics[width=0.6in]
                 {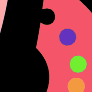}}

    \subfigure[] {\includegraphics[width=0.6in]
                 {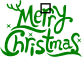}}
    \subfigure[] {\includegraphics[width=0.6in]
                 {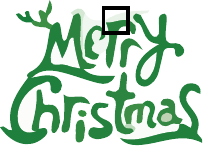}}
    \subfigure[] {\includegraphics[width=0.6in]
                 {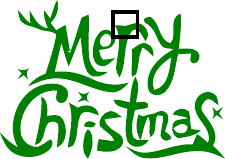}}
    \subfigure[] {\includegraphics[width=0.6in]
                 {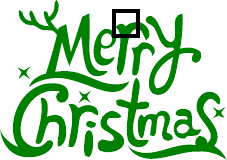}}
    \subfigure[] {\includegraphics[width=0.6in]
                 {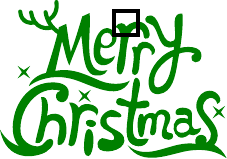}}

    \subfigure[] {\includegraphics[width=0.6in]
                 {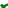}}
    \subfigure[] {\includegraphics[width=0.6in]
                 {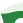}}
    \subfigure[] {\includegraphics[width=0.6in]
                 {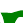}}
    \subfigure[] {\includegraphics[width=0.6in]
                 {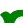}}
    \subfigure[] {\includegraphics[width=0.6in]
                 {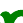}}

\caption{Comparisons in cases of obtuse corners. From left to right: raster
input, Adobe's result, Vector Magic's result, our result, ground truth.}
\label{fig:visual-comparison-corner}
\end{figure}

\textbf{Case 3: slightly bent edges.} Vectorizing various detailed structures,
such as slightly bent edges (Figures~\ref{fig:visual-comparison-bent}j and
\ref{fig:visual-comparison-bent}t), is also difficult for traditional
methods. Because the error associated with the generation of an intermediate
representation is unavoidable, small perturbations of the point chain are often
considered to be noise rather than signal. Therefore, certain slightly bent
edges in the resultant bezigons are straightened
(Figures~\ref{fig:visual-comparison-bent}g, \ref{fig:visual-comparison-bent}h
and \ref{fig:visual-comparison-bent}r). In our framework, we trust only the
original raster input and any prior knowledge regarding the curves. Although
not every type of structure can be preserved (e.g., large perturbations or
zigzag-like structures may be suppressed based on a priori knowledge of typical
vector images), small shapes and slightly bent edges are more likely to be
preserved in our results (Figures~\ref{fig:visual-comparison-bent}i and
\ref{fig:visual-comparison-bent}s).

\begin{figure}[!t]
\centering

    \subfigure[] {\includegraphics[width=0.6in]
                 {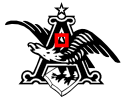}}
    \subfigure[] {\includegraphics[width=0.6in]
                 {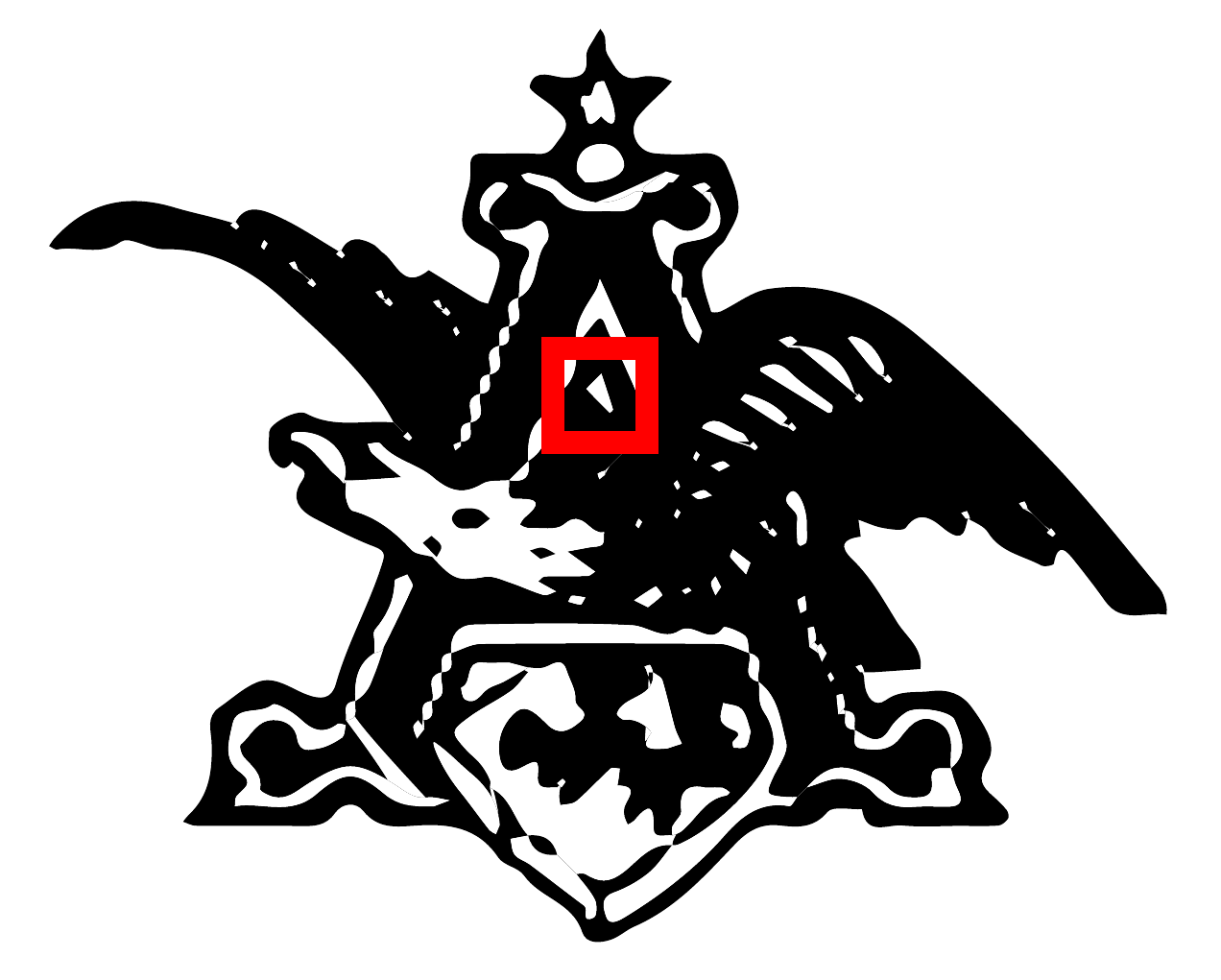}}
    \subfigure[] {\includegraphics[width=0.6in]
                 {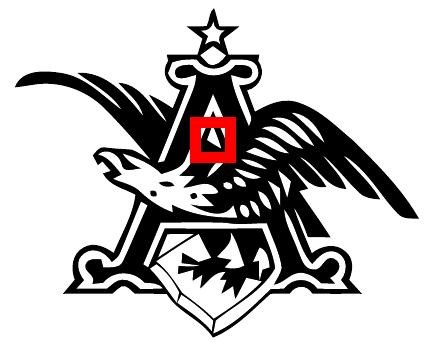}}
    \subfigure[] {\includegraphics[width=0.6in]
                 {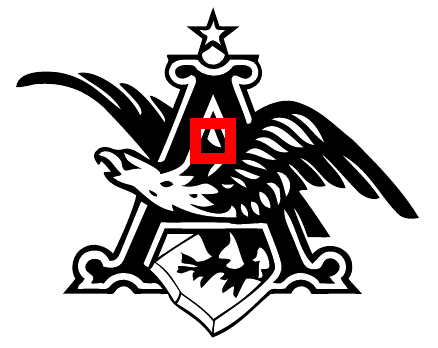}}
    \subfigure[] {\includegraphics[width=0.6in]
                 {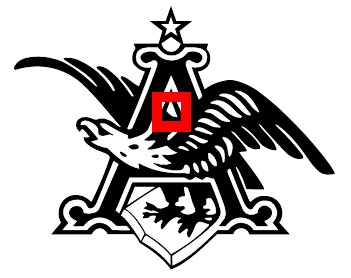}}

    \subfigure[] {\includegraphics[width=0.6in]
                 {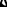}}
    \subfigure[] {\includegraphics[width=0.6in]
                 {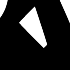}}
    \subfigure[] {\includegraphics[width=0.6in]
                 {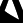}}
    \subfigure[] {\includegraphics[width=0.6in]
                 {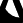}}
    \subfigure[] {\includegraphics[width=0.6in]
                 {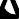}}

    \subfigure[] {\includegraphics[width=0.6in]
                 {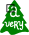}}
    \subfigure[] {\includegraphics[width=0.6in]
                 {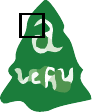}}
    \subfigure[] {\includegraphics[width=0.6in]
                 {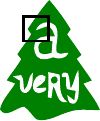}}
    \subfigure[] {\includegraphics[width=0.6in]
                 {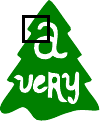}}
    \subfigure[] {\includegraphics[width=0.6in]
                 {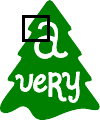}}

    \subfigure[] {\includegraphics[width=0.6in]
                 {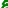}}
    \subfigure[] {\includegraphics[width=0.6in]
                 {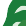}}
    \subfigure[] {\includegraphics[width=0.6in]
                 {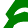}}
    \subfigure[] {\includegraphics[width=0.6in]
                 {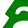}}
    \subfigure[] {\includegraphics[width=0.6in]
                 {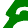}}

\caption{Comparisons in cases of slightly bent edges. From left to right:
raster input, Adobe's result, Vector Magic's result, our result, ground truth.}
\label{fig:visual-comparison-bent}
\end{figure}

From the analysis presented above, we can conclude that our direct bezigon
optimization for image vectorization produces more convincing vector results in
most cases.

\section{Conclusion and Future Work}
\label{sec:conclusion}

We have presented a novel framework for clipart image vectorization. In
contrast to other methods, the proposed approach optimizes bezigons by directly
observing the raster input and incorporating bezigon-based priors to minimize
the errors introduced by other intermediate procedures. Both quantitative and
qualitative comparisons demonstrate that the quality of the bezigons generated
by our approach is typically higher compared with those generated by the
current state-of-the-art method and by commonly used commercial software.

Of course, certain types of clipart images (e.g., noisy images or
low-resolution images that contain complex structures) exist that are too
ambiguous to be precisely vectorized by any automated approach, including our
method (Figure~\ref{fig:failure-results} shows such cases). Perhaps the best way to address these images is to incorporate a small
amount of user intervention. For this purpose, our system provides a friendly
graphical interface for user refinement during the course of bezigon
optimization. The evolving bezigons are presented in the interface. The user is
allowed to modify the location of any control point by dragging the mouse
cursor. Our system takes the modified bezigon as a new initial bezigon and
performs the subsequent optimization.

\begin{figure}[!t]
\centering

    \subfigure[] {\includegraphics[height=0.5in]
                 {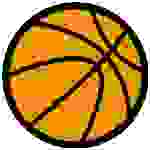}}
    \subfigure[] {\includegraphics[height=0.5in]
                 {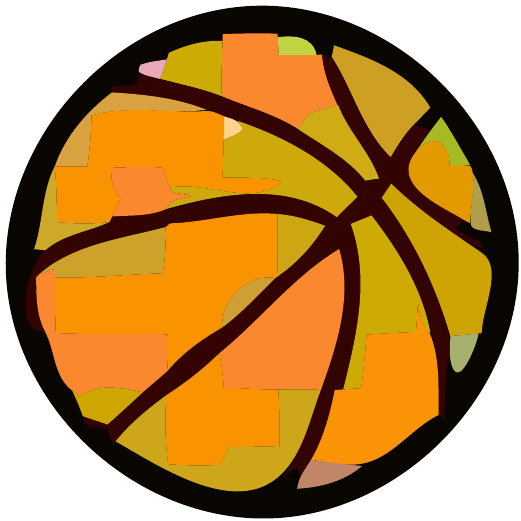}}
    \subfigure[] {\includegraphics[height=0.5in]
                 {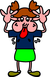}}
    \subfigure[] {\includegraphics[height=0.5in]
                 {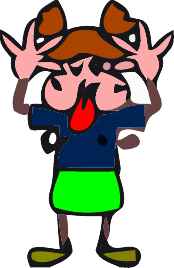}}

\caption{Cases producing poor results. (a) A noisy input image. (c) A low-resolution input image with complex content. (b)(d) Our vectorized output. }
\label{fig:failure-results}
\end{figure}

As future research, we will resolve additional ambiguities by incorporating
more prior knowledge regarding vector images for bezigon optimization. Because
we directly optimize the bezigons, it is trivial to incorporate such prior
information into our framework.

We also plan to develop a commercial software package based on the proposed
method. To make the software as efficient as possible, we will optimize the
code of the current implementation and consider parallelization of the proposed
approach. Remarkably, many components of our framework, ranging from the
wavelet rasterization to the optimization of local structures that are
completely irrelevant to each other, can be highly parallelized.

\ifCLASSOPTIONcompsoc
  \section*{Acknowledgments}
\else
  \section*{Acknowledgment}
\fi

This research project has been underway for nearly three years. We would like
to thank all the participants, especially two professional graphical designers,
Beck Yang and Fanco Ke, for their helpful suggestions and valuable comments.

\appendices
\label{sec:appendix}
\section{}

In this appendix we will introduce the complete definition of the rasterization function $R_{MS} (W;x,y)$, where $W$ are the parameter set of a bezigon, and give a proof to illustrate the continuity and differentiability of this function with respect to the geometrical parameters.

\subsection{Basic Definitions}
\label{subsec:appendix-definitions}

Before describing the rasterization function, we introduce some basic definitions that will be needed throughout this section.

Based on \cite{manson2011wavelet}, $R_{MS} (W;x,y)$ uses a hierarchical Haar wavelet representation to analytically calculate an anti-aliased raster image of a bezigon. Haar wavelets, as is well known, are represented by its mother wavelet function
\begin{equation}
\psi(t)=
\begin{cases}
1, & t \in [0,\frac{1}{2}], \\
-1, & t \in [\frac{1}{2},1), \\
0, & \text{otherwise}.
\end{cases}
\end{equation}
and its scaling function
\begin{equation}
\phi(t)=
\begin{cases}
1, & t \in [0,1), \\
0, & \text{otherwise}.
\end{cases}
\end{equation}

Based on the above two functions, the 1D Haar basis with a scaling parameter $s \in \mathbb{Z}$ and a translating parameter $l \in \mathbb{Z}$ could be formally defined as
\begin{eqnarray}
\psi_{s,k}(t) &=& \psi(2^s t - l), \qquad t \in \mathbb{R}, \\
\label{eq:6}
\phi_{s,k}(t) &=& \phi(2^s t - l), \qquad t \in \mathbb{R}.
\label{eq:7}
\end{eqnarray}

Now let $k=(k_x,k_y) \in \mathbb{Z}^2$, the 2D Haar basis defined as following will be used later:
\begin{eqnarray}
\psi_{s,k}^{(0,0)} (x,y) \!\!\!\! &=& \!\!\!\! 2^s \phi_{s,k_x} (x) \phi_{s,k_y} (y), \quad (x,y) \!\! \in \!\! \mathbb{R}^2, \\
\psi_{s,k}^{(0,1)} (x,y) \!\!\!\! &=& \!\!\!\! 2^s \phi_{s,k_x} (x) \psi_{s,k_y} (y), \quad (x,y) \!\! \in \!\! \mathbb{R}^2, \\
\psi_{s,k}^{(1,0)} (x,y) \!\!\!\! &=& \!\!\!\! 2^s \psi_{s,k_x} (x) \phi_{s,k_y} (y), \quad (x,y) \!\! \in \!\! \mathbb{R}^2, \\
\psi_{s,k}^{(1,1)} (x,y) \!\!\!\! &=& \!\!\!\! 2^s \psi_{s,k_x} (x) \psi_{s,k_y} (y), \quad (x,y) \!\! \in \!\! \mathbb{R}^2.
\end{eqnarray}

\subsection{Rasterization Function $R_{MS} (W;x,y)$ And Its Continuity}
\label{subsec:appendix-continuity}

According to \cite{manson2011wavelet}, the value of pixel $(x,y)$ in the raster image of a given 2D bezigon, indicated by the parameters $W=(B,C)$, takes the form
\begin{equation} % [TODO] 和前面重复了！
\begin{split}
&R_{MS}(W;x,y) = \\
&c(C;x,y) \sum_{j=1}^N
        \left\{\begin{split}
        \sum_{k \in K} c_{0,k}^{(0,0)} (B;j) \psi_{0,k}^{(0,0)} (x,y) \\
        + \sum_{s=0}^d \sum_{k \in K}
                \left[\begin{split}
                c_{s,k}^{(0,1)} (B;j) \psi_{s,k}^{(0,1)} (x,y) \\
                +c_{s,k}^{(1,0)} (B;j) \psi_{s,k}^{(1,0)} (x,y) \\
                +c_{s,k}^{(1,1)} (B;j) \psi_{s,k}^{(1,1)} (x,y)
                \end{split}\right]
        \end{split}\right\}, \\
&B \in \mathbb{R}^{6N}, C \in \mathbb{R}^3, (x,y) \in \Lambda, \\
&d \text{ is a given integer}, K \text{ is a finite set of } \mathbb{Z}^2.
\end{split}
\label{eq:12}
\end{equation}
Here, $c_{s,k}^{(\cdot)}(B;j)$ correspond to the wavelet coefficients contributed by the $j$-th B\'ezier curve segment:
\begin{equation}
\begin{split}
c_{s,k}^{(0,0)} (B;j) = & \int_0^1 2^s \tilde{\phi}_{s, k_x} (X_j(B_x;t)) \\ & \phi_{s, k_y} (Y_j(B_y;t)) Y'_j(B_y; t) dt, \\
c_{s,k}^{(0,1)} (B;j) = & \int_0^1-2^s \tilde{\psi}_{s, k_x} (Y_j(B_y;t)) \\ & \phi_{s, k_y} (X_j(B_x;t)) X'_j(B_x; t) dt, \\
c_{s,k}^{(1,0)} (B;j) = & \int_0^1 2^s \tilde{\psi}_{s, k_x} (X_j(B_x;t)) \\ & \phi_{s, k_y} (Y_j(B_y;t)) Y'_j(B_y; t) dt, \\
c_{s,k}^{(1,1)} (B;j) = & \int_0^1 2^s \tilde{\psi}_{s, k_x} (X_j(B_x;t)) \\ & \psi_{s, k_y} (Y_j(B_y;t)) Y'_j(B_y; t) dt.
\end{split}
\end{equation}

The notations $B_x, B_y$ and $X_j, Y_j (j=1,2,\hdots,N)$ are the same as Equation~\ref{eq:X_j-Y_j} and \ref{eq:B_x-B_y} in Section~\ref{sec:overview}. Note that given the bezigon parameters $B$, both $X_j$ and $Y_j$ are functions of one variable $t$, while both $X'_j$ and $Y'_j$ are first-order derivatives with respect to $t$. For all $s \in \mathbb{Z}$ and $l \in \mathbb{Z}$,
\begin{eqnarray}
\tilde{\phi}_{s,l}(t) = & \int_0^t \phi_{s,l} (u) du, \qquad t \in \mathbb{R}, \\
\label{eq:17}
\tilde{\psi}_{s,l}(t) = & \int_0^t \psi_{s,l} (u) du, \qquad t \in \mathbb{R}.
\label{eq:18}
\end{eqnarray}

It is obvious that both $\tilde{\phi}_{s,l}(t)$ and $\tilde{\psi}_{s,l}(t)$ are continuous with respect to the variable $t$ respectively. Also, if $t=h(B)$ is a continuous function of any parameters of $B$, both $\tilde{\phi}_{s,l}(t)$ and $\tilde{\psi}_{s,l}(t)$ are too. From Equation~\ref{eq:X_j-Y_j}, it is easy to see that both $X_j(B_x;t)$ and $Y_j(B_y;t)$ are continuous with respect to any parameters of $B_x$ and $B_y$. Therefore, $c_{s,k}^{(\cdot)} (B;j)$ are also continuous with respect to $B$. Thus the continuity of $R_{MS} (W;x,y)$ with respect to geometrical parameters $B$ is totally determined by above discussion and its formula~\ref{eq:12}. Such property is also reflected in Figure~\ref{fig:data-energy-comparison}, where the data energy function using $R_{MS} (W;x,y)$ is continuous with respect to an arbitrary geometrical parameter.

\subsection{Derivatives of $R_{MS} (W;x,y)$ with respect to geometrical parameters}
\label{subsec:appendix-derivatives}

We will show that $R_{MS} (W;x,y)$ is differentiable with respect to the geometrical parameters B, which verifies Theorem~\ref{theorem:differentiability} in Section~\ref{sec:directly-optimizing-bezigon}. Since the discontinuity of Haar function, the conclusion of Theorem~\ref{theorem:differentiability} is not obvious. To achieve this goal, we will use the theory of generalized functions and generalized derivatives \cite{gel1968generalized}. Following deductions are all in the sense of generalized function and generalized derivative.

We first express formally such derivatives as
\begin{equation}
\label{eq:20}
\begin{split}
&\frac{\partial R_{MS}(W;x,y)}{\partial x_{j,i}} = \\
&\sum_{j=1}^N
        \left\{\begin{split}
        \sum_{k \in K} \frac{\partial}{\partial x_{j,i}} c_{0,k}^{(0,0)} (B;j) \psi_{0,k}^{(0,0)} (x,y) \\
        + \sum_{s=0}^d \sum_{k \in K}
                \left[\begin{split}
                \frac{\partial}{\partial x_{j,i}} c_{s,k}^{(0,1)} (B;j) \psi_{s,k}^{(0,1)} (x,y) \\
                +\frac{\partial}{\partial x_{j,i}} c_{s,k}^{(1,0)} (B;j) \psi_{s,k}^{(1,0)} (x,y) \\
                +\frac{\partial}{\partial x_{j,i}} c_{s,k}^{(1,1)} (B;j) \psi_{s,k}^{(1,1)} (x,y)
                \end{split}\right]
        \end{split}\right\}, \\
&B \in \mathbb{R}^{6N},
\end{split}
\end{equation}
and
\begin{equation}
\label{eq:21}
\begin{split}
&\frac{\partial R_{MS}(W;x,y)}{\partial y_{j,i}} = \\
&\sum_{j=1}^N
        \left\{\begin{split}
        \sum_{k \in K} \frac{\partial}{\partial y_{j,i}} c_{0,k}^{(0,0)} (B;j) \psi_{0,k}^{(0,0)} (x,y) \\
        + \sum_{s=0}^d \sum_{k \in K}
                \left[\begin{split}
                \frac{\partial}{\partial y_{j,i}} c_{s,k}^{(0,1)} (B;j) \psi_{s,k}^{(0,1)} (x,y) \\
                +\frac{\partial}{\partial y_{j,i}} c_{s,k}^{(1,0)} (B;j) \psi_{s,k}^{(1,0)} (x,y) \\
                +\frac{\partial}{\partial y_{j,i}} c_{s,k}^{(1,1)} (B;j) \psi_{s,k}^{(1,1)} (x,y)
                \end{split}\right]
        \end{split}\right\}, \\
&B \in \mathbb{R}^{6N}
\end{split}
\end{equation}
for all $j=1,2,\hdots,N$, $i=1,2,3,4$, and $(x,y) \in \Lambda$.

Then the remaining problem is to discuss the differentiability of Haar basis coefficients with respect to geometrical parameters, i.e., the existence of $\frac{\partial c_{s,k}^{(\cdot)} (B;j) }{\partial x_{j,i}}$ and $\frac{\partial c_{s,k}^{(\cdot)} (B;j) }{\partial y_{j,i}}$ for all $j=1,2,\hdots,N$, $i=1,2,3,4$, $s=0,1,\hdots,d$, and $k \in K$.

\textbf{Generalized Derivatives of Haar Basis Functions.} It is well known that the generalized derivative of $\phi(t)$:
\begin{equation}
\phi'(t)=\delta(t)-\delta(t-1), \qquad t\in \mathbb{R}.
\end{equation}
Here $\delta$ is an impulse function satisfying:
\begin{equation}
\label{eq:23}
\int_{-\infty}^{\infty} \delta(t) f(t) dt = f(0).
\end{equation}
Here $f(t)$ is an arbitrary continuous function. Note that when composed with a continuous function $g(t)$, $\delta$ holds the following property \cite{gel1968generalized}:
\begin{equation}
\label{eq:24}
\delta(g(t)) = \sum_{t_i \in T} \frac{\delta(t-t_i)}{|g'(t_i)|}, \qquad t \in \mathbb{R}.
\end{equation}
Here $T$ is the set of the real roots of $g(t)$.
Similarly,
\begin{equation}
\psi'(t)= \delta(t)-2\delta(t-\frac{1}{2})+\delta(t-1), \qquad t \in \mathbb{R}.
\end{equation}

Therefore, for all $s \in \mathbb{Z}$, $l \in \mathbb{Z}$,
\begin{equation}
\label{eq:26}
\begin{split}
\phi'_{s,l}(t) &= \frac{d (\phi (2^st-l))}{dt} \\
&= 2^s[ \delta(2^st-l) - \delta(2^st-l-1) ]
\end{split}, \quad t \in \mathbb{R}.
\end{equation}
Similarly, for all $s \in \mathbb{Z}$, $l \in \mathbb{Z}$,
\begin{equation}
\begin{split}
\psi'_{s,l}(t) &= 2^s[ \delta(2^st-l) - 2\delta(2^st-l-\frac{1}{2}) \\ &+ \delta(2^st-l-1)) ] \\
& \quad B \in \mathbb{R}^{6N}.
\end{split}
\end{equation}

\textbf{Derivatives of Haar Basis Coefficients with Respect to Geometrical Parameters.} We first calculate $\frac{\partial c_{s,k}^{(0,0)}}{\partial x_{j,i}}$. According to the generalized functions theory \cite{gel1968generalized}, for all $j=1,2,\hdots,N$, $i=1,2,3,4$, $s=0,1,..,d$, $k_x \in K_x$ and $k_y \in K_y$,
\begin{equation}
\label{eq:28}
\begin{split}
\frac{\partial c_{s,k}^{(0,0)} (B;j)}{\partial x_{j,i}} = & \int_0^1 \frac{\partial}{\partial x_{j,i}} [ 2^s \tilde{\phi}_{s,k_x} (X_j(B_x;t)) \\
& \phi_{s,k_y} (Y_j(B_y;t)) Y'_j(B_y;t) ] dt \\
& \quad B \in \mathbb{R}^{6N}.
\end{split}
\end{equation}
Since $\phi_{s,k_y} (Y_j(B_y;t)) Y'_j(B_y;t)$ has nothing to do with the parameter $x_{j,i}$ according to Equation~\ref{eq:X_j-Y_j}, we have
\begin{equation}
\label{eq:29}
\begin{split}
\frac{\partial c_{s,k}^{(0,0)} (B;j)}{\partial x_{j,i}} =& 2^s \int_0^1 \frac{\partial \tilde{\phi}_{s,k_x} (X_j(B_x;t))}{\partial x_{j,i}} \\
& \phi_{s,k_y} (Y_j(B_y;t)) Y'_j(B_y;t) dt \\
=& 2^s \int_0^1 \phi_{s,k_x} (X_j(B_x;t)) \frac{\partial X_j(B_x;t)}{\partial x_{j,i}} \\
& \phi_{s,k_y} (Y_j(B_y;t)) Y'_j(B_y;t) dt \\
& \quad B \in \mathbb{R}^{6N}.
\end{split}
\end{equation}

Thus, the derivative of $c_{s,k}^{(0,0)} (B;j)$ with respect to any value of $x_{j,i}$ exists . Also, it can be analytically calculated by substituting Equation~\ref{eq:X_j-Y_j} and Equation~\ref{eq:7} into Equation~\ref{eq:29}.

Now we turn to $\frac{\partial c_{s,k}^{(0,0)}}{\partial y_{j,i}}$. Similar to Equation~\ref{eq:28} and \ref{eq:29}, we have
\begin{equation}
\label{eq:30}
\begin{split}
\frac{\partial c_{s,k}^{(0,0)} (B;j)}{\partial y_{j,i}} =& \int_0^1 \frac{\partial}{\partial y_{j,i}} [ 2^s \tilde{\phi}_{s,k_x} (X_j(B_x;t)) \\
& \phi_{s,k_y} (Y_j(B_y;t)) Y'_j(B_y;t) ] dt \\
=& 2^s \int_0^1 \tilde{\phi}_{s,k_x} (X_j(B_x;t)) \\
&  \frac{\partial }{\partial y_{j,i}} \left[ \phi_{s,k_y} (Y_j(B_y;t)) Y'_j(B_y;t) \right] dt \\
& \quad B \in \mathbb{R}^{6N},
\end{split}
\end{equation}
for all $j=1,2,\hdots,N$, $i=1,2,3,4$, $s=0,1,..,d$, $k_x \in K_x$ and $k_y \in K_y$.
Here
\begin{equation}
\begin{split}
& \frac{\partial }{\partial y_{j,i}} \left[ \phi_{s,k_y} (Y_j(B_y;t)) Y'_j(B_y;t) \right] \\
= & \frac{\partial \phi_{s,k_y} (Y_j(B_y;t))}{\partial y_{j,i}} Y'_j(B_y;t)
+  \phi_{s,k_y} (Y_j(B_y;t)) \frac{\partial Y'_j(B_y;t)}{\partial y_{j,i}}
\end{split}
\end{equation}
According to Equation~\ref{eq:24} and \ref{eq:26} we have:
\begin{equation}
\label{eq:31}
\begin{split}
& \frac{\partial \phi_{s,k_y} (Y_j(B_y;t))}{\partial y_{j,i}} Y'_j(B_y;t) \\
= & \phi'_{s,k} (Y_j(B_y;t)) Y'_j(B_y;t) \\
= & \left[\begin{split}
       & 2^s \delta (2^s Y_j (B_y;t) - k_y) \\
       - & 2^s \delta (2^s Y_j(B_y;t) - k_y -1)
    \end{split}\right]
 Y'_j (B_y;t) \frac{\partial Y_j(B_y;t)}{\partial y_{j,i}} \\
= & \sum_{t_0 \in T_0} \frac{2^s \delta (t-t_0)}{|2^s Y'_j (B_y; t_0)|} Y'_j(B_y;t_0) \frac{\partial Y_j(B_y;t_0)}{\partial y_{j,i}} \\
- & \sum_{t_1 \in T_1} \frac{2^s \delta (t-t_1)}{|2^s Y'_j (B_y; t_1)|} Y'_j(B_y;t_1) \frac{\partial Y_j(B_y;t_1)}{\partial y_{j,i}} \\
= & \sum_{t_0 \in T_0} \delta (t-t_0) \text{sgn} (Y'_j (B_y; t_0)) \frac{\partial Y_j(B_y;t_0)}{\partial y_{j,i}} \\
- & \sum_{t_1 \in T_1} \delta (t-t_1) \text{sgn} (Y'_j (B_y; t_1)) \frac{\partial Y_j(B_y;t_1)}{\partial y_{j,i}}, \\
& \quad B \in \mathbb{R}^{6N},
\end{split}
\end{equation}
for all $j=1,2,\hdots,N$, $i=1,2,3,4$, $s=0,1,..,d$, and $k_y \in K_y$.
Here $T_0$ and $T_1$ are the sets of the real roots of
\begin{equation}
g_1 (t) = 2^s Y_j (B_y;t)-k_y, \qquad t \in [0,1]
\end{equation}
and
\begin{equation}
g_2 (t) = 2^s Y_j (B_y;t)-k_y-1, \qquad t \in [0,1],
\end{equation}
respectively. Note that either $g_1(t)=0$ or $g_2(t)=0$ is a cubic equation in one variable (i.e., $t$). By substituting Equation~\ref{eq:31} into Equation~\ref{eq:30}, there is
\begin{equation}
\label{eq:32}
\begin{split}
\frac{\partial c_{s,k}^{(0,0)} (B;j)}{\partial y_{j,i}} = & 2^s
        [ \sum_{t_0 \in T_0} \int_0^1 \delta (t-t_0) \tilde{\phi}_{s,k_x} (X_j(B_x;t)) \\
        & \text{sgn} (Y'_j (B_y; t_0)) \frac{\partial Y_j(B_y;t_0)}{\partial y_{j,i}} dt \\
      - & \sum_{t_1 \in T_1} \int_1^1 \delta (t-t_1) \tilde{\phi}_{s,k_x} (X_j(B_x;t)) \\
        & \text{sgn} (Y'_j (B_y; t_1)) \frac{\partial Y_j(B_y;t_1)}{\partial y_{j,i}} dt \\
      + & \int_0^1 \tilde{\phi}_{s,k_x} (X_j(B_x;t))  \phi_{s,k_y} (Y_j(B_y;t))  \\
        & \frac{\partial Y'_j(B_y;t)}{\partial y_{j,i}} dt
    ] \\
    & \quad B \in \mathbb{R}^{6N},
\end{split}
\end{equation}
for all $j=1,2,\hdots,N$, $i=1,2,3,4$, $s=0,1,..,d$, $k_x \in K_x$ and $k_y \in K_y$.
From Equation~\ref{eq:23} we have
\begin{equation}
\int_0^1 \delta(t-u) f(t) dt = f(u), \qquad u \in (0,1).
\end{equation}
Therefore Equation~\ref{eq:32} could be written as
\begin{equation}
\label{eq:34}
\begin{split}
& \frac{\partial c_{s,k}^{(0,0)} (B;j)}{\partial y_{j,i}} \\
= & 2^s \!\!  \left[\begin{split}
        \!\!  & \! \sum_{t_0 \in T_0} \!\!  \tilde{\phi}_{s,k_x} (X_j(B_x;t_0)) \text{sgn} (Y'_j (B_y; t_0)) \frac{\partial Y_j(B_y;t_0)}{\partial y_{j,i}} \! \\
      - \!\!  & \! \sum_{t_1 \in T_1} \!\!  \tilde{\phi}_{s,k_x} (X_j(B_x;t_1)) \text{sgn} (Y'_j (B_y; t_1)) \frac{\partial Y_j(B_y;t_1)}{\partial y_{j,i}} \! \\
      + \!\!  & \int_0^1 \tilde{\phi}_{s,k_x} (X_j(B_x;t))  \phi_{s,k_y} (Y_j(B_y;t))  \frac{\partial Y'_j(B_y;t)}{\partial y_{j,i}} dt
    \end{split} \right] \\
    & \quad B \in \mathbb{R}^{6N},
\end{split}
\end{equation}
for all $j=1,2,\hdots,N$, $i=1,2,3,4$, $s=0,1,..,d$, $k_x \in K_x$ and $k_y \in K_y$.
Therefore the derivative of $c_{s,k}^{(0,0)} (B;j)$ with respect to $y_{j,i}$ exists. And it can be analytically calculated by substituting Equation~\ref{eq:X_j-Y_j}, Equation~\ref{eq:7} and Equation~\ref{eq:17} into Equation~\ref{eq:34}.

Similarly, for all $B \in \mathbb{R}^{6N}$, $j=1,2,\hdots,N$, $i=1,2,3,4$, $s=0,1,..,d$, $k_x \in K_x$ and $k_y \in K_y$, we can compute the remaining derivatives:
\begin{equation}
\begin{split}
& \frac{\partial c_{s,k}^{(0,1)} (B;j)}{\partial x_{j,i}} \\
= & 2^s \!\! \left[\begin{split}
      - \!\!  & \! \sum_{t_0 \in T_0} \!\!  \tilde{\psi}_{s,k_y} (Y_j(B_y;t_0)) \text{sgn} (X'_j (B_x; t_0)) \frac{\partial X_j(B_x;t_0)}{\partial x_{j,i}} \! \\
      + \!\!  & \! \sum_{t_1 \in T_1} \!\!  \tilde{\psi}_{s,k_y} (Y_j(B_y;t_1)) \text{sgn} (X'_j (B_x; t_1)) \frac{\partial X_j(B_x;t_1)}{\partial x_{j,i}} \! \\
      - \!\!  & \int_0^1 \tilde{\psi}_{s,k_y} (Y_j(B_y;t))  \phi_{s,k_x} (X_j(B_x;t))  \frac{\partial X'_j(B_x;t)}{\partial x_{j,i}} dt
    \end{split} \right],
\end{split}
\end{equation}
\begin{equation}
\begin{split}
\frac{\partial c_{s,k}^{(0,1)} (B;j)}{\partial y_{j,i}} =-2^s \int_0^1 & \psi_{s,k_y} (Y_j(B_y;t)) \frac{\partial Y_j(B_y;t)}{\partial y_{j,i}} \\
& \phi_{s,k_x} (X_j(B_x;t)) X'_j(B_x;t) dt,
\end{split}
\end{equation}
\begin{equation}
\begin{split}
\frac{\partial c_{s,k}^{(1,0)} (B;j)}{\partial x_{j,i}} =2^s \int_0^1 & \psi_{s,k_x} (X_j(B_x;t)) \frac{\partial X_j(B_x;t)}{\partial x_{j,i}} \\
& \phi_{s,k_y} (Y_j(B_y;t)) Y'_j(B_y;t) dt,
\end{split}
\end{equation}
\begin{equation}
\begin{split}
& \frac{\partial c_{s,k}^{(1,0)} (B;j)}{\partial y_{j,i}} \\
= & 2^s \!\! \left[\begin{split}
      - \!\!  & \! \sum_{t_0 \in T_0} \!\!  \tilde{\psi}_{s,k_x} (X_j(B_x;t_0)) \text{sgn} (Y'_j (B_y; t_0)) \frac{\partial Y_j(B_y;t_0)}{\partial y_{j,i}} \! \\
      + \!\!  & \! \sum_{t_1 \in T_1} \!\!  \tilde{\psi}_{s,k_x} (X_j(B_x;t_1)) \text{sgn} (Y'_j (B_y; t_1)) \frac{\partial Y_j(B_y;t_1)}{\partial y_{j,i}} \! \\
      - \!\!  & \int_0^1 \tilde{\psi}_{s,k_x} (X_j(B_x;t))  \phi_{s,k_x} (Y_j(B_y;t))  \frac{\partial Y'_j(B_y;t)}{\partial y_{j,i}} dt
    \end{split} \right],
\end{split}
\end{equation}
\begin{equation}
\begin{split}
\frac{\partial c_{s,k}^{(1,1)} (B;j)}{\partial x_{j,i}} =-2^s \int_0^1 & \psi_{s,k_x} (X_j(B_x;t)) \frac{\partial X_j(B_x;t)}{\partial x_{j,i}} \\
& \psi_{s,k_y} (Y_j(B_y;t)) Y'_j(B_y;t) dt,
\end{split}
\end{equation}
\begin{equation}
\label{eq:40}
\begin{split}
& \frac{\partial c_{s,k}^{(1,1)} (B;j)}{\partial y_{j,i}} \\
= & 2^s \!\! \left[\begin{split}
        \!\!  & \! \sum_{t_0 \in T_0} \!\!  \tilde{\psi}_{s,k_x} (X_j(B_x;t_0)) \text{sgn} (Y'_j (B_y; t_0)) \frac{\partial Y_j(B_y;t_0)}{\partial y_{j,i}} \! \\
      - 2\!\!  & \! \sum_{t_1 \in T_1} \!\!  \tilde{\psi}_{s,k_x} (X_j(B_x;t_1)) \text{sgn} (Y'_j (B_y; t_1)) \frac{\partial Y_j(B_y;t_1)}{\partial y_{j,i}} \! \\
      + \!\!  & \! \sum_{t_2 \in T_2} \!\!  \tilde{\psi}_{s,k_x} (X_j(B_x;t_2)) \text{sgn} (Y'_j (B_y; t_2)) \frac{\partial Y_j(B_y;t_2)}{\partial y_{j,i}} \! \\
      + \!\!  & \int_0^1 \tilde{\psi}_{s,k_x} (X_j(B_x;t))  \psi_{s,k_x} (Y_j(B_y;t))  \frac{\partial Y'_j(B_y;t)}{\partial y_{j,i}} dt
    \end{split} \right],
\end{split}
\end{equation}

Note that all these derivatives of Haar basis coefficients with respect to each geometrical parameter exist, and can be calculated analytically. Therefore, derivatives of the rasterization function $R_{MS} (W;x,y)$ with respect to the ge ometrical parameters could be also analytically calculated by substituting Equation~\ref{eq:29},\ref{eq:34}-\ref{eq:40} into Equation~\ref{eq:20} and Equation~\ref{eq:21} respectively.

Since there exist analytic derivatives for $R_{MS} (W;x,y)$ with respect to each geometrical parameter, the differentiability of $R_{MS} (W;x,y)$ is proved, which verifies Theorem~\ref{theorem:differentiability} in Section~\ref{sec:directly-optimizing-bezigon}.

\ifCLASSOPTIONcaptionsoff
  \newpage
\fi

\vfill


\end{document}